\documentclass{sig-alternate}

\usepackage{verbatim}
\usepackage{amssymb}
\usepackage{amsmath}
\usepackage[noend]{algpseudocode}
\usepackage{algorithm}

\usepackage{enumerate}

\newcommand{\vsigma}{\vec{\sigma}}			
\newcommand{\va}{\vec{a}}					
\newcommand{\hist}{{\cal H}}					
\newcommand{\act}{{\cal A}}					
\newcommand{\term}{{\cal Z}}					
\newcommand{\eu}[3]{\mathbb{E}^{#2}(u_{#1}|#3)}	
\newcommand{\pr}{P}						
\newcommand{\pra}[3]{\pr^{#2}(#1|#3)}			
\newcommand{\rnd}{\textit{rn}}					
\newcommand{\key}{\kappa}					
\newcommand{\stat}{\Theta}					
\newcommand{\id}{\textit{id}}					
\newcommand{\seed}{\textit{sd}}			         
\newcommand{\prng}{\textit{SG}}				
\newcommand{\miss}{\textit{Miss}}				
\newcommand{\seq}{\textit{seq}}				
\newcommand{\ids}{\{1 \ldots \nevt\}}			
\newcommand{\rec}{\textit{rec}}				
\newcommand{\se}{\textit{SE}}					
\newcommand{\re}{\textit{RE}}					
\newcommand{\pend}{\textit{PE}}				
\newcommand{\age}{\textit{age}}				
\newcommand{\false}{\mbox{\emph{False}}}		
\newcommand{\true}{\mbox{\emph{True}}}
\newcommand{\good}{\mbox{\emph{Good}}}
\newcommand{\bad}{\mbox{\emph{Bad}}}
\newcommand{\nrnd}{\tau} 		 	 		
\newcommand{\nrndm}{\tau^{\mbox{\tiny{M}}}}  	
\newcommand{\nrndd}{\tau^{\mbox{\tiny{D}}}}		
\newcommand{\nevt}{\nu}		                		 	
\newcommand{\dens}{\rho}                             		
\newcommand{\pset}{{\cal N}}                                 	
\newcommand{\evt}{{\cal E}}  	                               	
\newcommand{\nseq}{n^{\mbox{\tiny{S}}}}                 	
\newcommand{\npseq}{n^{\mbox{\tiny{E}}}}                
\newcommand{\pmon}{p^*}                				
\newcommand{\recev}[3]{\evt_{#1}^{#2}(#3)}		

\algblockdefx[Event]{UponEvent}{EndEvent}
   [1]{\textbf{Upon} #1}
   [0]{\textbf{End}}
\algblockdefx[EventP]{UponEventP}{EndEventP}
   [2]{\textbf{Upon} #1 $|$ #2}
   [0]{\textbf{End}}
\algblockdefx[Predicate]{Predicate}{EndPredicate}
   [2]{\textbf{Predicate} #1(#2)}
   [0]{\textbf{EndPredicate}}
\algblockdefx[After]{After}{EndAfter}
   [1]{\textbf{After} #1}
   [0]{\textbf{End}}

\newcommand{\fakeparagraph}[1]{\medskip\noindent\textbf{#1~}}

\newtheorem{theorem}{Theorem}[section]

\newtheorem{lemma}[theorem]{Lemma}
\newtheorem{proposition}[theorem]{Proposition}
\newtheorem{definition}[theorem]{Definition}

\conferenceinfo{ICDCN}{'15 Jan. 4-7, Goa, India}
\CopyrightYear{2015}
\crdata{978-1-4503-2928-6}

\begin{document}

\title{On the Range of Equilibria Utilities of a Repeated Epidemic Dissemination Game with a Mediator}

\numberofauthors{2}
\author{
\alignauthor
Xavier Vila\c{c}a\\
       \affaddr{INESC-ID, Instituto Superior T\'{e}cnico}\\
       \affaddr{Universidade de Lisboa, Portugal}\\
       \email{xvilaca@gsd.inesc-id.pt}
\alignauthor
Lu\'{i}s Rodrigues\\
       \affaddr{INESC-ID, Instituto Superior T\'{e}cnico}\\
       \affaddr{Universidade de Lisboa, Portugal}\\
       \email{ler@ist.utl.pt}
}

\maketitle
\begin{abstract}
We consider eager-push epidemic dissemination in a complete graph. Time is divided into synchronous stages. In each stage, 
a source disseminates $\nevt$ events. Each event is sent by the source, and forwarded by each node upon its first reception,
to $f$ nodes selected uniformly at random, where $f$ is the fanout. We use Game Theory to study the range of $f$ 
for which equilibria strategies exist, assuming that players are either rational or obedient to the protocol, and that they do not collude.
We model interactions as an infinitely repeated game. We devise a monitoring mechanism that extends the repeated game with communication rounds used for exchanging
monitoring information, and define strategies for this extended game. We assume the existence of a trusted mediator,
that players are computationally bounded such that they cannot break the cryptographic primitives used in our mechanism, 
and that symmetric ciphering is cheap. Under these assumptions, we show that, if the size of the stream is sufficiently large 
and players attribute enough value to future utilities, then the defined strategies are Sequential Equilibria of the extended game
for any value of $f$. Moreover, the utility provided to each player is arbitrarily close to that provided in the original game.
This shows that we can persuade rational nodes to follow a dissemination protocol that uses any fanout, while arbitrarily minimising the relative overhead of monitoring.
\end{abstract}




\section{Introduction}
This paper addresses the impact of rational behaviour in epidemic dissemination protocols, executed over a complete graph\,\cite{Birman:99}. 
An epidemic dissemination protocol operates as follows: a source splits a stream of bits into $\nevt$ events, which are sent to a set of $f$ nodes chosen
uniformly at random, where $f$ is known as the \emph{fanout}; nodes repeated this procedure upon the first reception of each event.
Protocols of this type are know to achieve a good tradeoff between high reliability of event delivery and communication overhead,
and have been used in a variety of applications such as video streaming\,\cite{Li:06,Li:08}. In this context, rational behaviour may be characterised
by the aim of maximising a utility, given as follows. Rational nodes value the stream, but they prefer to send as few messages as possible,
in order to spare bandwidth. Therefore, the utility can be defined as the difference between the benefits, which increase with the number of received events, 
and the communication costs for sending messages. This setting poses the problem that rational nodes always prefer not to forward any messages.

To address this issue, we explore the possibility of nodes interacting repeatedly in multiple executions of the dissemination protocol, in order to hold nodes
accountable for their present behaviour by adjusting their utility in the future. This models periodic streaming sessions (e.g., weekly sporting events).
Using Game Theory\,\cite{Osborne:94}, we study incentives to persuade rational nodes to follow the protocol, also admitting
the possibility that some nodes may be acquiescent, i.e., obedient to the protocol\,\cite{Aiyer:05,Wong:10}. We assume that rational nodes do not collude.
We model interactions as an infinitely repeated game where future utilities are discounted to the present by some factor $\delta \in (0,1)$,
which determines the value given by players to future utilities\footnote{We use the designations player and node to describe the entities of the system:
node refers to the operational side of the entity; player refers to the (rational) user controlling the node.}.
This is adequate when players are uncertain about the number of future interactions and value utilities obtained in the present
over future ones\,\cite{Osborne:94}. Our aim is to study the range of the values of $f$ used by dissemination protocols that correspond to equilibria strategies of the repeated game, that is, where
no player has any incentive to deviate from the protocol assuming that other players also do not deviate.

When players interact repeatedly, incentives may be based on \emph{direct} or \emph{indirect} reciprocity\,\cite{Trivers:71}.
With \emph{direct reciprocity}, each player $i$ adapts his strategy towards every other player $j$ in reaction to past actions of $j$ directed towards $i$ only. 
An example is the tit-for-tat strategy\,\cite{Axelrod:84,Cohen:03}. Most epidemic dissemination systems that cope 
with rational behaviour also use direct reciprocity\,\cite{Li:06,Li:08,Keidar:09}. Unfortunately, strategies of this type are vulnerable to the redundancy
of epidemic dissemination\,\cite{Vilaca:13}. Namely, if a player $i$ does not cooperate with $j$, then at best $j$ may punish $i$. 
However, when redundancy is high, $i$ has other neighbours from whom $i$ receives events. Hence, the impact of such punishment is 
arbitrarily low, since $i$ continues to receive events with sufficiently high probability. This makes direct reciprocity an ineffective type of incentives when $f$ is large.

\emph{Indirect reciprocity} circumvents the limitations of direct reciprocity by having nodes sharing information regarding private observations. 
This allows all the nodes to coordinate on effective punishments against any player $i$
by not forwarding any event to $i$, decreasing his utility to $0$. In Game Theoretical terms,
we can punish each player by decreasing his utility to the minimax value. Under this possibility, we might apply
a well known set of results called Folk Theorems, which state that we may devise equilibria strategies for infinitely repeated games 
that provide any feasible strictly positive utility to every player\,\cite{Mailath:07}, given that $\delta$ is sufficiently large.
In particular, these results imply that it is possible to sustain cooperation while using any fanout for disseminating events.
However, existing proofs of Folk Theorems assume that some underlying monitoring infra-structure provides information about the behaviour of each node,
at no cost to the players\,\cite{Mailath:07}. Such assumption is unrealistic, since any implementation of a monitoring mechanism always incurs communication costs.

\fakeparagraph{Goal}: We show that cooperation can be sustained for any fanout using a monitoring mechanism that is not free of cost. For this purpose, we prove the 
existence of equilibria strategies for the game induced by the monitoring mechanism. We consider the notion of Sequential Equilibrium\,\cite{Kreps:82},
which is stronger than Nash Equilibrium (NE) since it excludes strategies that rely on non-credible threats. For instance, when considering monitoring,
the notion of NE does not evaluate the optimality of a strategy when a player has observed a deviation and has to communicate this fact to other players.
In line with the Folk Theorems, we also aim at providing any feasible and strictly positive utility to each player, which requires the minimisation of 
the communication overhead of monitoring relative to the original dissemination protocol. Existing works have faced the challenge of implementing 
a distributed monitoring mechanism\,\cite{Guerraoui:10}, and performed a game theoretical analysis of epidemic dissemination\,\cite{Vilaca:13}. To the best of our knowledge,
none has studied the range of $f$ used by equilibria strategies of the repeated epidemic dissemination game.

\fakeparagraph{Challenges:}
To fulfil our goal, we define a monitoring mechanism that extends the infinitely repeated epidemic dissemination game with additional communication rounds,
and we propose a set of strategies for this game. In addition to disseminating a stream of events, nodes review the behaviour of every player and share
this information, which is used to decide when to punish each player. When a punishment against a player $i$ is in place,
no node sends events to $i$, denying any benefit to $i$. This way, the threat of punishment out-weighs the gain from deviating
from the specified strategy. The following main challenges are addressed:

\emph{Challenge 1: Strategic Monitoring}. Players may deviate when sharing monitoring information. 
For instance, a punishment of any player $i$ causes the overall reliability of dissemination to decrease. 
Hence, players are not willing to share information incriminating $i$.

\emph{Challenge 2: Mixed strategies}. Nodes randomly select the neighbours to forward each event. The difficulty lies in preventing
players from biasing this selection.

\emph{Challenge 3: Hidden Events}. While observing the actions of $i$, player $j$ does not observe the set of
events received by $i$. Thus, the monitoring mechanism may raise false positives, causing players to be punished undeservedly.

\emph{Challenge 4: Overhead of monitoring}. Nodes cannot share information regarding each disseminated event,
otherwise the overhead of monitoring is not minimised. Though, monitoring only a subset of disseminated events
introduces the problem of false negatives, where misbehaviour is undetected.

\fakeparagraph{Summary of Contributions:} 
We devise a monitoring mechanism and a set of strategies that are Sequential Equilibria of the extended game for any fanout,
provided that $\nevt$ and $\delta$ are sufficiently large. These strategies also minimise the communication overhead of monitoring
relative to the original dissemination protocol. It is important to notice that the overhead is minimised only relative to the total size of the stream.
We use symmetric cryptography to cope with strategic monitoring.
Assuming that symmetric ciphering costs are negligible and players are computationally bounded, this allows cheap punishments to be applied to players,
while creating incentives for them to continue forwarding events. Our results offer an improvement over existing work towards the goal of designing
a practical monitoring mechanism\,\cite{Guerraoui:10}, which does not consider strategic monitoring. We use a pseudo-random number to address the challenge raised by mixed strategies.   
We address the challenge of hidden events by relying on a trusted mediator to collect information about the events sent to and received from each node $i$.
We minimise monitoring overhead by having nodes reporting in expectation on only a sub-linear number of events.
The mediator coordinates the selection of these events to ensure that the probability of false negatives is sufficiently low, and that no false positives are raised.

\fakeparagraph{Paper Organisation:}
The remainder of the paper is organised as follows. In Section~\ref{sec:model}, we introduce our model.
Section~\ref{sec:mon} contains the monitoring mechanism and strategies.
In Section~\ref{sec:analysis}, we perform a Game Theoretical analysis of the strategies to prove the main result.
Section~\ref{sec:discuss} concludes the paper with a short discussion.

\section{Model}
\label{sec:model}
We consider a synchronous message passing system with reliable and authenticated communication.
The set of nodes is denoted by $\pset$, and $n$ is its cardinality. We consider that $\pset$ is common knowledge\footnote{Every player knows
this set, knows that every player knows this set, and so on.}.
Players do not collude and have perfect recall. Time is divided into stages, which are further divided into $\nrnd$ synchronous rounds.
In each stage $t$, the source disseminates a set $\evt^t$ of $\nevt$ events, drawn from a much larger but finite set $\evt$.
The process of generating $\evt^t$ must be sufficiently random, such that any player $i$ can guess any $e \in \evt^t$ beforehand
only with a sufficiently small probability. This is to justify the assumption that players prefer to receive these events than to try to guess them. Every event $e \in \evt^t$ 
has a unique identifier $\id \in \{1 \ldots \nevt\}$, and is disseminated as follows: first, the source sends $e$ to a subset of $f$ neighbours chosen
uniformly at random, where $f$ is the fanout. Then, each node forwards $e$ upon its first reception, also to a subset of $f$ neighbours
selected uniformly at random. This process ends until no node forwards $e$ or until a maximum delay of $\dens$ rounds to deliver $e$
is reached, after which the event is said to expire. We assume that $\dens$ does not increase with $\nevt$.
A different event is introduced in each round by the source. We let $\nrnd \geq \nevt + \dens$ to allow every event to be disseminated
until it expires. We consider that the source is obedient to the protocol and may act as a trusted mediator. Using the language
of~\,\cite{Aiyer:05,Wong:10}, we say that the source is acquiescent. We treat every other player as rational, although our results still hold if
other acquiescent players exist. In Section~\ref{sec:discuss}, we discuss how to distribute the role of the mediator.
Players are assumed to be computationally bounded - they cannot break the cryptographic primitives used in our strategies, in the time required by a stage.

Interactions are modelled as an infinitely repeated epidemic dissemination game.
An action of player $i$ is a vector $\va_i \in \act_i$ specifying for each $j \in \pset \setminus \{i\}$ a message $a_i(j)$ sent by $i$ to $j$. Messages contains a finite
number of tuples $(\id,e)$, where $\id \in \ids$ and $e \in \evt$. A history $h \in \hist$ is a finite sequence of action profiles $\va \in \act$ specifying all the messages sent after multiple rounds.
Any player $i$ cannot completely observe a history; instead, $i$ only observes a corresponding private history $h_i \in \hist_i$, specifying the messages sent
and received by $i$ in $h$. We say that histories $h_i$ and $h$ are from stage $t$ if the first round following these histories belongs to $t$.
A strategy $\sigma_i \in \Sigma_i$ specifies a probability distribution $\sigma_i(.|h_i)$ over the actions taken by $i$ in the round immediately
succeeding the observation of $h_i$. A strategy profile $\vsigma \in \Sigma$ specifies the strategy followed by every player. Given any $\vsigma \in \Sigma$,
each player $i$ forms a belief $\mu$ regarding the realised history $h \in \hist$ after the observation of $h_i \in \hist_i$,
in the form of a probability $\mu^{\vsigma}(h|h_i)$. $\mu$ is also common knowledge. If we fix some strategy profile $\vsigma$
that is expected to be followed by any player, then in any realised history $h \in \hist$ the behaviour of an acquiescent player 
$i$ is always compatible with $i$ having followed $\vsigma$ in $h$. Henceforth, whenever referring to a history $h \in \hist$
and strategy profile $\vsigma \in \Sigma$, we will implicitly consider that the behaviour of the source (and the mediator) in $h$
is compatible with $\vsigma$. Thus, we never analyse histories where acquiescent entities have deviated in the past, since this analysis
is irrelevant to the proof of equilibrium. This fact is also captured by the defined belief system: for any $h_i \in \hist_i$ and $h \in \hist$,
we have $\mu^{\vsigma}(h|h_i) > 0$ only if all acquiescent entities have followed $\vsigma$ in $h$. This implies that $i$ never observes
a deviation of some acquiescent entity in any $h_i$.

\subsection{Utility}
\label{sec:exp-util}
The expected utility of the infinitely repeated game obtained by any player $i$ depends on the following factors:
i) the private history $h_i \in \hist_i$ observed by $i$, initially equal to $\emptyset$; ii) the strategy profile $\vsigma$  followed by
every player; iii) the belief system $\mu$; and iv) the realised utility $u_i$ of receiving events and sending messages during each stage following the observation of $h_i$.
The pair $(\vsigma,\mu)$ allows $i$ to form an expectation of what occurred in the past,
in terms of a probability distribution over the histories $h \in \hist$. Given any $h$, $i$ can predict
the future behaviour of any player $j$ in any future stage $t$, given that $j$ follows $\sigma_j$.
More precisely, $\vsigma$ defines a probability distribution over the outcomes $\term^t \subseteq \hist$ of stage $t$,
where each outcome $z \in \term^t$ is a history that matches the end of stage $t$. We denote by $\pra{z}{\vsigma}{h}$
the probability of $z$ being reached after $h$, given that players follow $\vsigma$. The realised utility for stage $t$ is a function $u_i(z)$ of
the outcome $z \in \term^t$ reached in stage $t$. This function quantifies the average benefits of receiving events and the average costs of sending messages,
per disseminated event. More precisely, we consider that every player receives a benefit $\beta$ per received event $e \in \evt^t$ and incurs a cost $\alpha$ per bit
sent in a message. $u_i(z)$ is given by the total benefits minus the total costs. We divide $u_i(z)$ by $\nevt$ in order to normalise it to
the average utility per disseminated event. We provide a formal definition of $u_i(z)$ in Section~\ref{sec:realised:util}, after defining the monitoring mechanism.

With this in mind, the expected utility for stage $t$ given the observation of $h_i$, denoted by $\eu{i}{\vsigma,\mu,t}{h_i}$,
is the weighted sum over every history $h$ compatible with the observation of $h_i$ and outcomes of stage $t$ following $h$:
\begin{equation}
\eu{i}{\vsigma,\mu,t}{h_i} = \sum_{h \in \hist}\sum_{z \in \term^t} \mu^{\vsigma}(h|h_i) \pra{z}{\vsigma}{h} u_i(z).
\end{equation}

Finally, for any $i \in \pset$ and $h_i \in \hist_i$ from any stage $t$, $\eu{i}{\vsigma,\mu}{h_i}$ is the weighted infinite sum over every stage $t' \geq t$ of the expected utility of
$t'$. We use a discount factor $\delta \in (0,1)$ to discount future utilities to the present:
\begin{equation}
\eu{i}{\vsigma,\mu}{h_i} = \sum_{t' \geq t} \delta^{t'-t} \eu{i}{\vsigma,\mu,t'}{h_i}.
\end{equation}

\subsection{Notion of Equilibrium}
We consider the notion of Sequential Equilibrium (SE)\,\cite{Kreps:82}. We say that a pair
$(\vsigma,\mu)$ is a SE if it is Sequentially Rational and Consistent. $(\vsigma,\mu)$ is Sequentially Rational if 
$\sigma_i$ maximises $\eu{i}{\vsigma,\mu}{h_i}$ for any $h_i \in \hist_i$, conditional on the belief that other players follow $\vsigma$ after $h_i$ is observed.
Formally, let $(\sigma_i',\vsigma_{-i})$ be the strategy profile where every $j \in \pset \setminus \{i\}$ follows $\sigma_j \in \vsigma_{-i}$ and $i$ follows $\sigma_i'$.

\begin{definition}
The pair $(\vsigma,\mu)$ is Sequentially Rational iff for every $i \in \pset$, $h_i \in \hist_i$, and $\sigma_i' \in \Sigma_i$, it holds:
$$\eu{i}{\vsigma,\mu}{h_i} \geq \eu{i}{(\sigma_i',\vsigma),\mu}{h_i}.$$
\end{definition}

We do not include the formal definition of Consistency (c.f.~\cite{Kreps:82}). Informally, $(\vsigma,\mu)$ is Consistent if $\mu$ is defined using the 
Bayes rule according to the behaviour specified by $\vsigma$, whenever possible.
When some $h_i$ is observed that is inconsistent with the hypothesis that players have been following $\vsigma$,
the definition of consistency requires the specification of an alternative hypothesis for explaining the observed behaviour.
We fix a belief system $\mu$ that suits our purposes, defined as follows.
For any $\epsilon > 0$, define $\vsigma^\epsilon$ as the strategy profile where every $i$ follows $\sigma_i(.|h_i)$ after each $h_i \in \hist_i$ with probability $1-\epsilon$, 
and, with probability $\epsilon$, $i$ follows any available action with positive probability. 
Since every history is consistent with players following $\vsigma^{\epsilon}$, we can apply the Bayes rule
to completely define $\mu^{\vsigma^{\epsilon}}$. Then, we set $\mu^{\vsigma}(.|h_i) = \lim_{\epsilon \to 0} \mu^{\vsigma^{\epsilon}}(.|h_i)$. 
The intuition is that the observed behaviour is explained by players following $\vsigma$ and occasionally making
mistakes with a small probability.

The above definition of Sequentially Rational pair $(\vsigma,\mu)$ is problematic, since it requires the analysis of all possible alternative strategies to that specified by $\vsigma$.
Fortunately, we can simplify this task by analysing only local deviations according to the One-deviation Property\,\cite{Hendon:96}.
A local deviation for player $i$ after private history $h_i \in \hist_i$ is an action that is not prescribed by $\sigma_i$ with positive probability,
while every player still follows $\vsigma$ after $i$ observes $h_i$, and $i$ follows $\sigma_i$ in every round following the deviation.
Formally, let $\eu{i}{\vsigma,\mu}{h_i,\va_i}$ denote the expected utility of $i$ when every player follows $\vsigma$ after $h_i$ is observed, except only
that $i$ follows $\va_i$ immediately after observing $h_i$.

\begin{proposition}
\textbf{One-deviation}\,\cite{Hendon:96}. The pair $(\vsigma,\mu)$ is Sequentially Rational iff for every $i \in \pset$, $h_i \in \hist_i$, and
$\va_i,\va_i' \in \act_i$ such that $\sigma_i(\va_i|h_i) > 0$, it holds:
$$\eu{i}{\vsigma,\mu}{h_i,\va_i} \geq \eu{i}{\vsigma,\mu}{h_i,\va_i}.$$
\end{proposition}

\subsection{Central Claim}
In line with Folk Theorems\,\cite{Mailath:07}, we aim at defining equilibria strategies for any fanout $f$, while providing to each player any strictly
positive expected utility of the dissemination game as the average utility of the repeated game. The average utility is computed as $(1-\delta)\eu{i}{\vsigma,\mu}{\emptyset}$, where $\emptyset$
is the initial empty history. When all players forward events using a fanout $f$, the expected utility for any stage is $\bar{u}(f) = q(f)(\beta - \gamma f)$, where $q(f)$ is the probability of a given node receiving each event,
$\beta$ is the benefit per event, and $\gamma$ is the cost for forwarding the event. Folk Theorems imply that, if $\beta > \gamma f$ and $\delta$ is sufficiently large, 
then an equilibrium strategy for the infinitely repeated game exists that yields $\bar{u}(f)$ to every player as the average utility of the repeated game.
In this work, we aim to prove a slightly weaker result. First, we can only ensure the existence of such strategy if $\beta > c \gamma f$
for some constant $c > 0$ that may be greater than $1$\,\footnote{With appropriate optimisations, we can get $c \leq 3$.}. Second, we need both $\delta$ and $\nevt$ to be sufficiently large.
Finally, due to the overhead of monitoring, we can only provide to each player a utility arbitrarily close but never exactly
equal to $\bar{u}(f)$. Theorem~\ref{theorem:folk} formalises the central claim.

\begin{theorem}
\label{theorem:folk}
Fix any fanout $f$ and constant $\epsilon > 0$. There exist a monitoring mechanism, a strategy profile $\vsigma$ for the dissemination game, and a belief system $\mu$,
and there exist constants $c > 0$, $\bar{\delta} \in (0,1)$, and $\bar{\nevt} \in \mathbb{N}$, such that, if $\beta > c \gamma f$,
then, for every $\delta \in (\bar{\delta},1)$ and $\nevt > \bar{\nevt}$, $(\vsigma,\mu)$ is a SE and, for every $i \in \pset$:
$$|(1-\delta)\eu{i}{\vsigma,\mu}{\emptyset} - \bar{u}(f)| < \epsilon.$$
\end{theorem}

The core of the proof of this theorem is the definition of a monitoring mechanism executed every stage that, combined with a dissemination strategy, provides a SE. Also, 
its costs are sub-linear on $\nevt$. By dividing by $\nevt$, we obtain the average monitoring costs per event. For an arbitrarily large $\nevt$,  the impact of monitoring in the expected utility is arbitrarily small.
Alternatively, we could fix $\nevt$ and minimise the overhead of monitoring by only executing the mechanism once every period of $t$ stages. Here, we could arbitrarily increase $t$ to decrease the impact of monitoring.
This alternative approach would be almost identical to the one used in this paper, both in terms of the definition of the monitoring mechanism and the main arguments 
of the proof, albeit its greater complexity. We discuss this possibility in Section~\ref{sec:discuss}.

\section{Monitoring Mechanism}
\label{sec:mon}
We extend the dissemination game by adding $\nrndm$ monitoring rounds to each stage. For convenience, these rounds are added at the beginning
of the stage. Therefore, a stage is now divided into $\nrnd = \nrndm + \nrndd$ rounds: the first $\nrndm$ rounds are used for exchanging
monitoring information; the last $\nrndd$ rounds are used to disseminate events as in the dissemination game. 
Recall that $\dens$ is the maximum delay to deliver any event. We still set $\nrndd \geq \nevt+ \dens$ 
to allow every event to be disseminated until it expires. The trusted mediator is
responsible for collecting monitoring information from each node. This consists in \emph{accusations} and \emph{reports}.
An accusation flags a deviation of some player, whereas a report, which is relative to some $\id \in \ids$, indicates for each pair of nodes $(i,j)$
whether $i$ sent to/received from $j$ a tuple containing $\id$.
For each $f > 0$, we define a strategy profile $\vsigma^f$, and fix a belief $\mu^f$ that is Consistent with $\vsigma^f$ as defined in~\cite{Kreps:82}.
The main incentive for players to not deviate is based on indirect reciprocity.
The mediator collects accusations and reports in stage $t$ relative to the behaviour of each player $i$ in stage $t-1$, and gives a verdict on whether
$i$ should be punished during stage $t$. If $i$ is punished, then $i$ is still obliged to incur costs in stage $t$; 
if $i$ deviates again, then his punishment is extended to stage $t+1$. We address the main challenges as follows:

\emph{Strategic Monitoring}: We enforce the following two properties. First, each player $i$ is only allowed to send monitoring information
relative to other nodes; information sent by $i$ has no effect on the probability of $i$ being punished.
Second, regardless of the punishments being applied, the reliability of dissemination remains constant. We achieve this using commutative symmetric ciphering. 
For each node $i$, the mediator sends a key $\key_i$ to every node $j \neq i$, used by $j$ to cipher events sent to $i$. While being punished, $i$ is unable to
retrieve the disseminated events. However, $i$ still forwards each ciphered event normally. Every node $j \neq i$ receiving
such event is aware that $i$ is being punished and is capable of retrieving the original event,  given his knowledge of $\key_i$.

\emph{Mixed Strategies}: In dissemination rounds, every node $i$ forwards each received event to a set $S$ of $f$ nodes. 
To prevent players from biasing the selection of $S$, $S$ is specified by a pseudo-random number generator (PRNG), seeded by a random
number $\seed_i$, sent by the mediator to $i$ only. 

\emph{Hidden Events}: For each pair of nodes $(i,j)$ and $\id \in \ids$, $i$ reports to the mediator the round
when $i$ first received from/sent to $j$ a tuple containing $\id$. Given these reports, the mediator is able to determine whether $i$ received $\id$
and forwarded it as expected. More precisely, the mediator verifies whether $i$ forwarded $\id$ as specified by PRNG 
and $\seed_i$, immediately after its first reception. For each identifier $\id$ not received by $i$, the mediator verifies if $i$ did not send a tuple with identifier
$\id$. If any of the above conditions fails, then the mediator triggers a punishment of $i$.

\emph{Monitoring Overhead}: Nodes report only on a subset of events selected by the mediator in a non-deterministic fashion.
The expected number of reported events is sub-linear on $\nevt$, such that the relative overhead of monitoring decreases as $\nevt$ increases.
More precisely, we split the identifier space $\ids$ into $\nseq$ sequences, each containing $\npseq$ different identifiers. We define $\nrndm = 2\nseq + 2$.
In the first monitoring round, nodes send accusations to the mediator. In each even round $r < \nrndm$, for each node $i$, the mediator decides
whether to monitor the identifiers from the sequence corresponding to $r$, with independent probability $\pmon$. If so, then the mediator sends
a notification to every node $j \neq i$, in which case $j$ must report, in round $r+1$, on the identifiers from the corresponding sequence.
In the last round, the mediator sends the seeds, the keys, and the verdicts on each player. These parameters are specified in Section~\ref{sec:res}.

We now describe the main components of the monitoring mechanism in more detail.

\subsection{Symmetric Cipher}
Let ${\cal K}$ be the set of keys. For any key $\key \in {\cal K}$, let $(v)_\key$ be the operation that 
returns datum $v$ ciphered with $\key$. We will consider primitives for which the ciphering operation is the same as that of deciphering.
Hence, we consider that it holds $((v)_\key)_{\key} = v$. We need the following two properties to be fulfilled:

\vspace{0.05in}
\textit{Integrity}. For every $\key \in {\cal K}$ and $e \in \evt$, $(e)_\key \in \evt$.

\vspace{0.05in}
\textit{Commutativity}. For every $\key,\key' \in {\cal K}$, $((v)_\key)_{\key'} = ((v)_{\key'})_{\key}$.  
\vspace{0.05in}

These properties hold when using any stream cipher, for which the ciphering operation consists in applying an \emph{xor} between the datum
and a stream of bits generated from the key. We need Integrity to ensure that players only send valid dissemination messages, regardless of the punishments.
Commutativity is required for the scenarios when multiple punishments are applied simultaneously. More precisely, we need every player to be able
to obtain disseminated events while not being punished. Moreover, we need every node to punish other nodes while being punished.
This raises the following possibility. Consider nodes $1$ to $4$, and suppose that exactly both nodes $2$ and $3$ are being punished. Consider the following sequence of 
dissemination steps of event $e$: node $1$ sends $e^1 = (e)_{\key_2}$ to node $2$; node $2$ sends $e^2 = (e^1)_{\key_3}$ to node $3$; and node $3$ sends an event $e'$ to node $4$.
Since $4$ is not being punished, we need to ensure that he is able to retrieve $e$ from $e'$. If $3$ were to send $((e)_{\key_2})_{\key_3}$, then this would
only be possible if $4$ knew that the event followed the path $1,2,3$. Instead, $3$ sends $e' = (e^2)_{\key_2}$. By Commutativity, 
it holds $e' = (e)_{\key_3}$. Since $4$ knows $\key_3$, he only needs to know that $3$ is being punished in order to be able to retrieve $e$ from $e'$.

With this in mind, a node being punished may receive any event $e$ ciphered with $\key_i$, and possibly ciphered with $\key_j$ for some $j \neq i$.
In order for punishments to be effective, $i$ must not be able to retrieve $e$ when it has no access to $\key_i$, and has not received $e$ in plain,
regardless of what dissemination messages $i$ may have received in the past. For this purpose, we need to ensure that: 
1) $i$ never receives $\key_i$, which is true by construction of our strategy; and 2) $\key_i \neq \key_j$ for every $j \in \pset \setminus \{i\}$,
ensuring that $((e)_{\key_i})_{\key_j} \neq e$ for any $e \in \evt$. In addition, we need the following property to hold.
Intuitively, before receiving any tuple with identifier $\id$, $i$ can guess the corresponding event $e \in \evt^t$ with a small probability $p$.
Non-disclosure ensures that $i$ does not gain information that allows him to guess $e$ with a probability significantly higher than $p$,
provided that $i$ only receives $e$ ciphered with $\key_i$, and possibly with $\key_j$ for some $j \in \pset \setminus \{i\}$.

Formally, fix any $i \in \pset$, $h_i \in \hist_i$ from any stage $t$ and $\id \in \ids$. Let $\recev{i}{t}{h_i,\id} \subseteq \evt$ be the set of events 
such that for any $e \in \recev{i}{t}{h_i,\id}$ there exists $j \in \pset \setminus \{i\}$ and dissemination round $r$ from stage $t$ preceding the observation of $h_i$ such that
$j$ sends $(\id,e)$ to $i$ in round $r$. We say that $i$ has not received $e \in \evt^t$ in $h_i$ in plain iff for $e' = (e)_{\key_i}$:
$$\recev{i}{t}{h_i,\id} \subseteq \{e'\} \cup \{(e')_{\key_j} | j \in \pset \setminus \{i\}\}.$$

Let $\pra{e}{}{h_i}$ be the probability of $i$ guessing $e$ after the observation of $h_i$. We need the following property to hold:

\vspace{0.05in}
\textit{Non-disclosure}. For any $\epsilon > 0$, there exists a symmetric ciphering primitive such that, for any player $i \in \pset$, private history $h_i \in \hist_i$ from stage $t \in \mathbb{N}$,
and event $e \in \evt^t$ disseminated with identifier $\id \in \ids$, if $i$ has not received $e$ in plain in $h_i$, then for any $h_i' \in \hist_i$ from stage $t$ that fulfils $\recev{i}{t}{h_i',\id} = \emptyset$
and $\recev{i}{t}{h_i',\id'} = \recev{i}{t}{h_i,\id'}$ for every $\id' \in \ids \setminus \{\id\}$, it holds:
$$\pra{e}{}{h_i} \leq \pra{e}{}{h_i'} + \epsilon.$$
\vspace{0.02in}

These properties can be fulfilled by using a block ciphering algorithm with the CTR mode of operation\,\cite{Bellare:97,iso:06}. 
A careful selection of keys is also required. To simplify, we have the mediator selecting a new unique key per node in each stage, generated uniformly at random.

\subsection{Pseudo-Random Number Generator}
We assume the existence of a function $\prng_i^f$ per node $i$, defined as follows. Given a seed $\seed_i$, this function
returns a sequence of subsets, where $\prng_i^f(\seed_i,\id) \subseteq \pset \setminus \{i\}$ is the of $f$ nodes to whom $i$ must forward
the event with identifier $\id$. We need $\prng_i^f$ to fulfil the following requirements. Assuming that $\seed_i$ is chosen uniformly at random,
the probability of $\prng_i^f(\seed_i,\id)$ returning any subset $S$ of $f$ nodes is arbitrarily close to the probability $p$ of $i$ selecting $S$ uniformly
at random. Second, we need the stronger requirement of conditional independence.
Namely, if we fix the subsets generated by $\prng_i^f$ for any identifier other than $\id$, then the probability of $\prng_i^f(\seed_i,\id)$
returning any subset is still arbitrarily close to $p$. This is formalised by the following property, where $\pra{\vec{S}}{\prng_i^f}{\vec{S}_{-\id}}$
is the probability of $\prng_i^f(\seed_i,\id)$ returning a sequence of subsets $\vec{S}$ given that the seeds are selected uniformly at random among those
that yield any sequence $\vec{S}'$ such that $S_{\id'} = S_{\id'}'$ for every $\id' \in \ids \setminus \{\id\}$, independently of the value of $S_{\id}'$.

\vspace{0.05in}
\emph{PRNG1}. Fix any node $i \in \pset$ and sequence of subsets $\vec{S} = (S_{\id})_{\id \in \ids}$ with $S_{\id} \subseteq \pset \setminus \{i\}$ and $|S_{\id}| = f$ for every $\id \in \ids$.
For any constant $\xi > 0$, there exists a function $\prng_i^f$ such that for every $\id \in \ids$:
$$\left|\pra{\vec{S}}{\prng_i^f}{\vec{S}_{-\id}} - \frac{1}{\binom{n-1}{f}}\right| \leq \xi.$$
\vspace{0.05in}

This property can be fulfilled by a PRNG function such as the one used in~\cite{Abraham:13}. Additional details are included in~Appendix~\ref{sec:prng}.

\subsection{Strategy}
It is useful to define $\vsigma^f$ using a state machine representation\,\cite{Osborne:94}.
Each history $h_i \in \hist_i$ is mapped into a private state of each player $i$. Transitions between states occur when 
players follow any action profile $\va \in \act$. Given any $h_i \in \hist_i$, $\sigma_i^f(.|h_i)$ suggests a probability distribution 
over the set of actions available to $i$ after the observation of $h_i$. $i$ may or may not follow this suggestion, 
such that a state may be reached where $i$ knows he has deviated in the past. The only exception is acquiescent players,
who never deviate. We now specify the state, transition rules, and strategy for every $i \in \pset$, including the source.

\subsubsection{State}
Let $\rnd \in \{1 \ldots \nrnd\}$ be the round number. Node $i$ keeps for each node $j$ a variable $\stat_i(j) \in \{\good,\bad\}$, where $\stat_i(j) = \good$
iff $i$ has observed only valid actions from $j$. In addition, $i$ keeps a set $\miss_i$ of sequences of identifiers that contain identifiers that were not forwarded appropriately
according to $\prng_i^f$. For each node $j$, $i$ keeps two sets $\re_i(j)$ and $\se_i(j)$ of tuples $(\id,r)$, representing events with identifier $\id$ received from and sent to $j$
in round $r$, respectively. Finally, $\pend_i$ contains tuples $(\id,e)$ with an identifier $\id$ and an event $e$ to be forwarded.
An important aspect is that this state is finite, implying that memory is bounded.

\subsubsection{Transition Rules}
Algorithm~\ref{alg:update} contains the pseudo-code.
Every node initialises $\stat_i(j) = \good$ at the beginning of each stage and sets $\stat_i(j) = \bad$ when $j$
does not follow a valid action (Lines~\ref{line:update-stat} and~\ref{line:update-stat2}).
Valid actions are enumerated as follows. In monitoring rounds, each message has
a fixed size. In the first monitoring round, node $i$ must send to the mediator for each $j \neq i$
a message the size of an accusation against $j$. In any even monitoring round, for each node $j \neq i$ such that the mediator
requested the corresponding sequence of identifiers, $i$ must send two tuples per identifier $\id$ from the requested sequence.
If sufficient information is not available, then $i$ must send padding.
In a dissemination round, $i$ is allowed to send any set of tuples $(\id,e)$, but only one tuple per identifier. Moreover, the identifier must not have expired and must have already been
introduced. More precisely, define $\age(\id,r) = r + 1 - (\id + \nrndm)$. Event with identifier $\id$ is introduced in round $r$ such that $\age(\id,r) = 1$, i.e., $r = \nrndm+\id$.
A tuple containing $\id$ and sent in round $r$ is valid iff $\age(\id,r) \in \{1 \ldots \dens\}$. This restriction ensures that in any round every player has
to forward at most $\dens$ tuples to $f$ nodes. To simplify, we assume that every node has sufficient bandwidth to send $\dens f$ tuples in a single round.

Node $i$ registers in $\re_i(j)$ and $\se_i(j)$ for each $\id \in \ids$ the round number when $j$ first sent to and received from $i$ a tuple containing $\id$, 
in a valid dissemination message (Lines~\ref{line:add-se} and~\ref{line:add-re}). When $i$ receives $\id$ for the first time and 
$\id$ has not expired ($\age(\id,\rnd) \leq \dens$), $i$ selects a tuple $(\id,e)$ to be forwarded,
sent by some node $j$ chosen according to some deterministic rule such as the node with smallest identifier (Lines~\ref{line:store-event}-\ref{line:store-event:end}).
Not every deterministic rule is allowed, since we need to ensure that in any two scenarios 1 and 2 where the set of nodes $S$ sending $\id$ to $i$ in scenario 1
is a subset of scenario 2, and $j \in S$ is selected in 2, then $j$ is also selected in 1. An alternative valid rule would be the node with the largest identifier.
If $j$ is being punished, then $i$ forwards $(\id,(e)_{\key_j})$ instead (Line~\ref{line:save-cipher}).
Let $\seq(\id)$ denote the sequence to which identifier $\id$ belongs to.
$i$ adds $\seq(\id)$ to $\miss_i$ whenever $i$ knows that, if the mediator requests reports relative to $\seq(\id)$, 
then $i$ will be punished (Line~\ref{line:add-seq}). This occurs exactly when the reports relative to $\id$ are inconsistent: $i$ sends a tuple containing $\id$ 
to $j \notin \prng_i(\seed_i,\id)$ (Line~\ref{line:bad-node}), $i$ fails to send $\id$ to $j \in \prng_i(\seed_i,\id)$ immediately after the first reception of $\id$ (Line~\ref{line:missing-node}),
or $i$ sends $\id$ prior to receiving it (Line~\ref{line:premature}). Notice that we update $\re_i$ after updating
$\miss_i$, such that identifiers first received in the present round do not count as being inconsistent.

\begin{figure}[!t]
\begin{algorithm}[H]
\caption{Transition Rules}
\label{alg:update} 
{
\scriptsize
\begin{algorithmic}[1]
\After{first monitoring round}
	\State{$\rnd$ $\gets$ $1$}
	\State{$\pend_i,\miss_i$ $\gets$ $\emptyset$}
	\ForAll{$j \in \pset$}
		\State{$\re_i(j),\se_i(j)$ $\gets$ $\emptyset$}
		\State{$\stat_i(j)$ $\gets$ $\good$}\label{line:init-stat}
	\EndFor
	\State{Set $\stat_i(i) = \bad$ if $i$ sent invalid message}
\EndAfter

\Statex

\After{any other monitoring round}
	\State{$\rnd$ $\gets$ $\rnd+1$}
	\ForAll{$j \in \pset$}
		\State{Update $\stat_i(j)$}\label{line:update-stat}
	\EndFor
\EndAfter

\Statex

\After{dissemination round}
	\State{$\rnd$ $\gets$ $\rnd+1$}
	\ForAll{$j \in \pset$}
		\State{Update $\stat_i(j)$}\label{line:update-stat2}
	\EndFor
	\ForAll{valid $\id$ sent by $i$ to $j$ for the first time}
		\State{Add $(\id,\rnd)$ to $\se_i(j)$}\label{line:add-se}
	\EndFor
	\ForAll{$\id \in \ids$: inconsistent($\id$, $\se_i$, $\re_i$)}
		\State{Add $\seq(\id)$ to $\miss_i$}\label{line:add-seq}
	\EndFor
	\ForAll{valid $\id$ sent by $j$ to $i$ for the first time}
		\State{Add $(\id,\rnd)$ to $\re_i(j)$}\label{line:add-re}
	\EndFor	
	\ForAll{$\id \in \ids$ received for the first time}
		\If{$\age(\id,\rnd)\leq\dens$}\label{line:store-event}
			\State{$j$ $\gets$ node with smallest identifier to send $\id$}
			\State{$e$ $\gets$ event such that $j$ sent $(\id,e)$}
			\If{$i$ received accusation against $j$}
				\State{$e$ $\gets$ $(e)_{\key_j}$}\label{line:save-cipher}
			\EndIf
			\State{Add $(\id,e)$ to $\pend_i$}
		\EndIf\label{line:store-event:end}
	\EndFor
\EndAfter

\Statex

\Predicate{inconsistent}{$\id$, $S$, $R$}
	\If{$\exists_{j \notin \prng_i^f(\seed_i,\id)}\exists_r$ $(\id,r) \in S(j)$}\label{line:bad-node}
		\State{\Return $\true$}
	\ElsIf{$\exists_{j, r}$ $(\id,r) \in R(j)$}
		\State{$r$ $\gets$ minimum round when $i$ received $\id$ in $R(j)$}
		\If{$\age(\id,r) \in \{1 \ldots \dens-1\}$}
			\If{$\exists_{l \in \prng_i^f(\seed_i,\id)} (\id,r+1) \notin S(j)$}\label{line:missing-node}
				\State{\Return $\true$}
			\EndIf
		\EndIf
	\ElsIf{$\exists_{j,r}$ $(\id,r) \in S(j)$}\label{line:premature}
		\State{\Return $\true$}
	\Else
		\State{\Return $\false$}
	\EndIf
\EndPredicate
\end{algorithmic}
}
\end{algorithm}
\end{figure}

\subsubsection{Strategy Definition}
The pseudo-code is included in Algorithm~\ref{alg:strat}.
Player $i$ stops sending messages once $\stat_i(i) = \bad$, since a punishment in the next stage is inevitable. 
In the first monitoring round, $i$ sends an accusation against $j \neq i$ iff $\stat_i(j) = \bad$ at the end of the previous stage (Line~\ref{line:send-acc}). 
In an even monitoring round, the mediator notifies $i$ for each node $j \neq i$ whether $i$ should report 
on the sequence $S$ of identifiers corresponding to the current round, with probability $\pmon$ (Line~\ref{line:request}).
When $i$ receives this notification, $i$ sends for each identifier $\id \in S$ any tuples $(\id,r) \in \re_i(j)$ and $(\id,r') \in \se_i(j)$ (Lines~\ref{line:report}-\ref{line:report:end}).
In the last monitoring round, the mediator notifies every node $l \neq j$ that $j$ must be punished when some node $k \neq j$ has sent an accusation against $j$,
or he detects an inconsistency regarding some $\id \in \ids$, according to the reports $\re_j$ and $\se_j$ regarding events received and sent by $j$, respectively (Lines~\ref{line:accuse}-\ref{line:accuse:end}). 
The mediator also sends $\seed_i$ to $i$ and $\key_i$ to every $j \neq i$, ensuring that $\key_i$ is unique.
In a dissemination round, $i$ sends any tuple $(\id,e) \in \pend_i$ to every $j \in \prng_i(\seed_i,\id)$, ciphering $e$ with $\key_j$ if $j$ is being punished (Lines~\ref{line:fwd}-\ref{line:fwd:end}).
Notice that $i$ only sends a tuple $(\id,e)$ if $\seq(\id) \notin \miss_i$. This is because once $i$ fails to send some identifier from sequence $\seq(\id)$, 
sending $\id$ as specified by $\prng_i$ does not affect the probability of being punished in the next stage.
Thus, it is optimal to drop every tuple with an identifier from that sequence. Later, we specify how to define $\pmon$ such that
for every identifier $\id$ with $\seq(\id) \notin \miss_i$ it is still optimal for $i$ to forward $\id$.

\begin{figure}[!t]
\begin{algorithm}[H]
\caption{Strategy}
\label{alg:strat} 
{
\scriptsize
\begin{algorithmic}[1]
\UponEvent{first monitoring round of stage $t>1$}
	\ForAll{$j \in \pset$}
		\If{$\stat_i(j) = \bad$}
			\State{Send accusation against $j$ to mediator}\label{line:send-acc}
		\EndIf
		\State{Send padding}
	\EndFor
\EndEvent

\Statex

\UponEventP{even monitoring round $\rnd < \nrndm$}{$i$ is the mediator}
	\ForAll{$j \in \pset$}
		\State{With probability $\pmon$, send $(j,\textit{Monitor})$ to every $l \neq j$}\label{line:request}
	\EndFor
\EndEventP

\Statex

\UponEventP{odd monitoring round}{$1<\rnd<\nrndm$ and $\stat_i(i) = \good$}
	\ForAll{$j \in \pset \setminus \{i\}$: mediator sent $(j,\textit{Monitor})$}\label{line:report}
		\State{$S_j$ $\gets$ current sequence of identifiers}
		\State{Send all tuples $(\id,r) \in \re_i(j)$ with $\id \in S_j$}
		\State{Send all tuples $(\id,r) \in \se_i(j)$ with $\id \in S_j$}
		\State{Send padding}
	\EndFor\label{line:report:end}	
\EndEventP

\Statex

\UponEventP{last monitoring round $\nrndm$}{$i$ is mediator}
	\ForAll{$j \in \pset$}\label{line:accuse}
		\If{$\exists_{l \neq j}$ $l$ sent accusation against $j$}
			\State{Send accusation against $j$ to every $o \in \pset$}
		\ElsIf{$\exists_{\id}$ inconsistent($\id$, $\se_j$, $\re_j$)}
			\State{Send accusation against $j$ to every $l \in \pset$}
		\EndIf
		\State{$\seed_j$ $\gets$ random value}
		\State{Send $\seed_j$ to $j$}
		\State{$\key_j$ $\gets$ unique random key}
		\State{Send $\key_j$ to every $l \neq j$}
	\EndFor\label{line:accuse:end}
\EndEventP

\Statex

\UponEventP{dissemination round}{$\stat_i(i) = \good$}
	\ForAll{$(\id,e)\in \pend_i: \seq(\id) \notin \miss_i$}\label{line:fwd}
		\ForAll{$j \in \prng_i(\seed_i,\id)$}
			\If{$i$ has accusation against $j$}
				\State{$e$ $\gets$ $(e)_{\key_j}$}
			\EndIf
			\State{Send $(\id,e)$ to $j$}
		\EndFor
	\EndFor\label{line:fwd:end}
\EndEventP

\Statex

\UponEventP{dissemination of event $e \in \evt^t$}{$i$ is the source}
	\State{$\id$ $\gets$ $\rnd - \nrndm + 1$}
	\State{Add $(\id,e)$ to $\pend_i$}\label{line:dissem}
\EndEventP
\end{algorithmic}
}
\end{algorithm}
\end{figure}

\subsection{Realised Utility}
\label{sec:realised:util}
We now define $u_i(z)$ for every outcome $z \in \term^t$ of any stage $t$. We consider that every player $i$ incurs a fixed cost $\alpha$ per bit
sent in any message. In addition, $i$ obtains a benefit $\beta$ per event $e \in \evt^t$ received by $i$ during stage $t$. We denote the set of received tuples by $\rec_i(z)$, where for every $(\id,e) \in \rec_i(z)$
we have $e \in \evt^t$. Two issues arise when trying to define $\rec_i(z)$: 1) the exact definition of $i$ receiving $(\id,e)$ and 2) the reception of ciphered events.

Regarding the first issue, since we are considering reliable communication channels, $i$ receives $(\id,e)$ iff some node $j \neq i$ sends $(\id,e)$ to $i$.
Thus, we may consider that $i$ obtains a benefit $\beta$ in this case.
However, given our definition of strategies, it is possible that $i$ sends $(\id,e)$ to some node before receiving it in the first place, and then this tuple loops back to $i$.
Such behaviour increases the expected utility of $i$. Since $i$ is unlikely to guess $e$ beforehand, we consider that $i$ values $e$ 
iff $i$ knows that it could not have been introduced by $i$. More precisely, $i$ receives a benefit $\beta$ per disseminated tuple $(\id,e)$ with $e \in \evt^t$ iff some $j \neq i$ sends $(\id,e)$ to $i$ in
round $r$ of stage $t$, and $i$ has not sent $(\id,e)$ to any node in any round $r' < r$.

Regarding the second issue, it is possible that $i$ never receives some $e \in \evt^t$ in plain, but instead ciphered
with the key $\key_j$ of some node $j \neq i$. We define the strategy in a way that, if $e' = (e)_{\key_j}$ and $i$ receives $e'$, then it is because $j$ is being punished,
$j$ has sent $e'$ to $i$, and $i$ knows that $j$ is being punished. Hence, $i$ is able to compute $(e')_{\key_j}$, in which case we say that $i$ 
can \emph{retrieve} $e$ from $e'$. Naturally, $i$ can retrieve $e$ whenever $i$ receives $e$ in plain. In addition, it may be possible that a history is reached
where $i$ receives $e$ ciphered with multiple different keys and $i$ is still able to retrieve $e$. To generalise this intuition,
we say that $i$ is able to \emph{retrieve} $e$ from $e'$ whenever $i$ can perform some computation over $e'$, given the history of interactions with other nodes,
in order to obtain $e$. $i$ can retrieve $e$ from $e'$ when $e=e'$. Also, if $i$ can retrieve $e$ from $e'$, $i$ receives $e''=(e')_{\key_j}$ from $j$,
and $i$ knows that $j$ is being punished with key $\key_j$, then $i$ can retrieve $e$ from $e''$. By Non-disclosure, if some $j$ sends $(\id,e')$ to $i$ and $e' = (e)_{\key_i}$
or $e' = ((e)_{\key_i})_{\key_j}$, then $i$ cannot retrieve $e$ from $e'$. We do not make any further assumptions regarding when $i$ can retrieve $e$ from $e'$.

With this in mind, we consider that, for any $e \in \evt^t$ disseminated by the source with identifier $\id$, we have $(\id,e) \in \rec_i(z)$ iff there exists a round $r$ from stage $t$ and node $j \in \pset \setminus \{i\}$
such that $j$ sends $(\id,e')$ to $i$ in round $r$, $i$ can retrieve $e$ from $e'$, and $i$ has not sent $(\id,e')$ to any node in any round $r'< r$ from stage $t$.
This leads to the following definition of realised utility. Recall that we normalise the total benefits and costs to the average per disseminated event, by dividing
it by $\nevt$. Let $|z_i^{t,r}(j)|$ represent the size of the message sent by $i$ to $j \in \pset$ in round $r$ of stage $t$:
\begin{equation}
u_i(z) = \frac{1}{\nevt} \left( \beta \cdot |\rec_i(z)| - \alpha \cdot \sum_{r =1}^{\nrnd} \sum_{j \in \pset \setminus \{i\}}|z_i^{t,r}(j)| \right).
\end{equation}

\section{Analysis}
\label{sec:analysis}
The analysis is divided into three parts. First, we show correctness properties of $\vsigma^f$ for any $f > 0$.
Then, we compare the utility of following $\vsigma^f$ with that of deviating. We conclude with the proof of the main result.

We use the notation $v^h$ and $v^{h_i}$ to denote the value of state variable $v$ after any histories $h$ and $h_i$ are realised, respectively.
In this context, we say that a history $h'$ \emph{succeeds} $h$ if $h'$ is reached with positive probability when players follow $\vsigma^f$ after $h$; 
$h'$ \emph{immediately succeeds} $h$ if $h'$ is reached one round after $h$. Similarly, $h'$ precedes $h$ if $h'$ is a starting sub-sequence of $h$; 
$h'$ immediately precedes $h$ if $\rnd^h = \rnd^{h'}+1$.  

The proof relies the following facts regarding any node $i$ and history $h$ from stage $t$: i) for any $z \in \term^t$ succeeding $h$,
$\stat_i^z(i) = \stat_i^h(i)$ and $\miss_i^z = \miss_i^h$; ii) $\stat_i^h(i) = \good$ if and only if, for every $j \in \pset$ and $\id \in \ids$, $\stat_j^h(i) = \good$;
and iii) for any sequence of identifiers $s$, we have $s \in \miss_i^h$ if and only if there exists $\id \in \ids$ such that
$s=\seq(\id)$ and the reports relative to $i$ are inconsistent with regard to $\id$, such that if every player follows $\vsigma^f$
in the future, then $i$ is punished if and only if the mediator requests the identifiers from a sequence in $\miss_i^h$.

The analysis relies on the following proposition, which follows by construction of the strategy.  
\begin{proposition}
\label{prop:facts}
For any history $h \in \hist$ from stage $t \in \mathbb{N}$ and node $i \in \pset$, the following hold: 1) for any $z \in \term^t$ succeeding $h$,
$\stat_i^z(i) = \stat_i^h(i)$ and $\miss_i^z = \miss_i^h$; 2) $\stat_i^h(i) = \good$ iff, for every $j \in \pset \setminus \{i\}$, $\stat_j^h(i) = \good$;
and 3) for any identifier $\id \in \ids$, we have $\id \in \miss_i^h$ iff there exists $\id' \in \seq(\id)$ such that the reports relative to $i$ are inconsistent with regard to $\id'$.
\end{proposition}
\begin{proof}
Fix $t$, $h$, and $i$. Fix $z$ succeeding $h$. $i$ only changes $\stat_i(i)$ from $\good$ to $\bad$. Thus, if $\stat_i^h(i) = \bad$, then $\stat_i^z(i) = \bad$. Otherwise, $i$ only changes
$\stat_i^h(i)$ after $h$ if $i$ sends an invalid message, which never occurs while $i$ follows $\sigma_i^f$. Therefore, $\stat_i^h(i) = \stat_i^z(i)$. By construction, while following $\sigma_i^f$,
$i$ forwards every tuple with identifier $\id$ and $\seq(\id) \notin \miss_i^h$, first received in round $r$, to exactly all the nodes from $\prng_i^f(\seed_i,\id)$ in round $r+1$.
Thus, $i$ never adds any new sequence to $\miss_i^h$ while following $\sigma_i^j$, implying that $\miss_i^h = \miss_i^z$. This proves 1).
Now, recall that every $j \in \pset$ initialises $\stat_j(i) = \good$ at the beginning of stage $t$. If $\stat_i^h(i) = \bad$, then $i$ sent an invalid message to some $j \neq i$,
implying that $\stat_j^h(i) = \bad$. Otherwise, no $j$ updated $\stat_j(i)$ to $\bad$. This proves 2). Finally, since $i$ initialises $\miss_i = \emptyset$ in stage $t$,
we have $\id \in \miss_i^h$ iff there exists $\id' \in \seq(\id)$ such that $\textit{inconsistent}(\id',\se_i^h,\re_i^h)$ is true. By construction, this holds iff 
the reports relative to $i$ are inconsistent with regard to $\id'$. This proves 3).
\end{proof}

\subsection{Correctness}
We show two sets of properties regarding monitoring and dissemination. Monitoring properties characterise the probability of any node $i \in \pset$
being punished in the present or future stages as a function of the current state.
Dissemination properties quantify the probability of any node $i$ receiving events, already disseminated or disseminated only in the future,
as a function of the current state and the present action $\va_i$ of $i$. This allows us to compute the expected utility of any player $i$ for each possible action he takes in the present,
and this way prove that our strategies are SE by applying the One-deviation Property.

Lemma~\ref{lemma:mon-corr} enumerates monitoring Properties~M1-M4 valid for any player $i$ and history $h$, assuming that every player follows $\vsigma^f$ after $h$.
M1-M3 state that the probability of $i$ being punished in the next stage is a function of $\stat_i^h(i)$ and $\miss_i^h$, while $i$ is never punished in 
future stages other than the next. In addition, Property~M4 shows that $i$ cannot influence the punishments being applied to him in the present
stage by only deviating in the current round.

\begin{lemma}
\label{lemma:mon-corr}
For any history $h \in \hist$ from stage $t \in \mathbb{N}$ and player $i \in \pset$, the following properties hold:

\vspace{0.05in}
\textbf{M1.} If $\stat_i^h(i) = \good$, then the mediator accuses $i$ in stage $t+1$ with probability $\pr(|\miss_i^h|) = 1-(1-\pmon)^{|\miss_i^h|}$.
  
\vspace{0.05in}
\textbf{M2.} If $\stat_i^h(i) = \bad$, then the mediator accuses $i$ in stage $t+1$ with probability $1$.

\vspace{0.05in}
\textbf{M3.} The mediator never accuses $i$ in any stage $t' > t+1$.

\vspace{0.05in}
\textbf{M4.} For any two actions $\va_i,\va_i' \in \act_i$ and round $r> \rnd^{h}$ from stage $t$, $i$ is punished in round $r$ either by every node or by no node,
the probability of $i$ being punished is the same after $\va_i^*$ and $\va_i'$, and every node $j \in \pset \setminus \{i\}$ gets the same key $\key_i$.
\end{lemma}
\begin{proof}
\textbf{M1.} Fix any $z \in \term^t$ succeeding $h$. By Proposition~\ref{prop:facts}, we have $\stat_i^z(i) = \good$ and $\miss_i^z = \miss_i^h$, 
and $\stat_j^z(i) = \good$ for every $j \neq i$. Thus, no node $j$ sends an accusation against $i$ to the mediator
in the first round of stage $t+1$. Moreover, since nodes never remove entries from $\re$ and $\se$,
for every sequence $S \in \miss_i^h$, reports relative to $i$ after $z$ are inconsistent regarding some $\id \in S$.
If the mediator selects $S$ in the corresponding even monitoring round of stage $t+1$, then every node sends his reports relative to $i$ and regarding every event from $S$, including $\id$.
Thus, the mediator detects the inconsistency regarding $\id$ and emits an accusation against $i$ in the last monitoring round of stage $t+1$.
This occurs with probability $\pr(|\miss_i^z(i)|) = \pr(|\miss_i^h(i)|)$.

\vspace{0.05in}
\textbf{M2.} Fix any $z \in \term^t$ succeeding $h$. By Proposition~\ref{prop:facts}, it holds $\stat_j^z(i) = \bad$ for some $j \neq i$. Hence, $j$ sends an 
accusation against $i$ to the mediator in the first monitoring round of stage $t+1$, and the mediator accuses $i$ in the last
monitoring round of stage $t+1$.

\vspace{0.05in}
\textbf{M3.} Let $t' \geq t+1$ be any stage number and fix any $z \in \term^{t'-1}$ succeeding $h$. Since $i$ follows $\vsigma^f$ in the first monitoring round of stage $t'$, 
we have $\stat_i^{h'}(i) = \good$ and $\miss_i^{h'} = \emptyset$ for $h'$ immediately succeeding $z$. By M1, it follows that the mediator accuses $i$ in stage $t'+1$ with probability:
$$\pr(|\miss_i^{h'}|) = \pr(0) =0.$$
  
\vspace{0.05in}
\textbf{M4.} Regarding the first and last statements, recall that the mediator is trusted. Hence, we are focusing on any history $h$ where he has not deviated in the past.
Therefore, whether $h$ precedes the last monitoring round, the mediator always sends the same key $\key_i$ to every $j \in \pset \setminus \{i\}$,
and either sends an accusation to every node or to none. Regarding the second statement, the mediator ignores any information sent by $i$ when determining whether $i$ should be punished.
If $h \in \term^{t-1}$, then the probability of $i$ being accused by the mediator depends solely on $\stat_j(i)$, $\rec_j(i)$, and $\se_j(i)$ for every $j \in \pset \setminus \{i\}$, and on
the probability $\pmon$ of the mediator selecting each sequence of identifiers. Otherwise, the probability of $i$ being punished
depends on the same information except $\stat_j(i)$, and depends on what $j$ already sent to the mediator in previous monitoring rounds of stage $t$.
If $h$ follows the last monitoring round, then either every node $j \in \pset \setminus \{i\}$ already has an accusation against $i$ in $h$,
or none has, regardless of the present action.
\end{proof}

Lemma~\ref{lemma:dissem-corr1} enumerates properties D1-D5. The complete proof is in Appendix~\ref{app:corr}.
Given any history $h$ from stage $t$, event $e$, and node $i \in \pset$, let $g^h(e,i) = (e)_{\key_i}$ 
if the mediator triggers a punishment of $i$ in stage $t$ with key $\key_i$, or $g^h(e,i) = e$ otherwise.
By~M4, we have that for any $j \neq i$, if $(\id,e') \in \pend_j$ and $j$ forwards $\id$ to $i$,
then $j$ sends $g^h(e',i)$. Conversely, if $j$ receives $(\id,e')$ for the first time from $i$, then $j$ adds $(\id,g^h(e',i))$ to $\pend_j$.
For any history $h$ and two actions $\va_i^*$ and $\va_i'$, we use the notation $v_i^*(i)$ and $v_i'(i)$ to 
represent the value of any state variable $v_i$ resulting from $i$ following $\va_i^*$ and $\va_i'$ after $h$, respectively.
For some $j$ and message $a_i(j)$, we consider that $\id \in a_i(j)$ for some $j$ if there exists a tuple $(\id,e) \in a_i(j)$.

Properties D1-D3 refer to events introduced only in the future. Namely, Property D1 states that, regardless of the punishments being applied, 
a player $i$ obtains a disseminated tuple $(\id,e)$ iff $i$ is not punished and receives $(\id,g^h(e,j))$ from some $j$.
D2 and D3 indicate that $i$ cannot influence the probability of receiving any event not yet introduced.
Properties D4-D5 refer to events being disseminated. More precisely, D4 states that if $i$ follows two alternative actions in which $i$ sends a given identifier $\id$ already disseminated
to the same set of nodes, then $i$ receives a tuple containing $\id$ after both actions with the same probability.
D5 states that if the source previously disseminated a tuple $(\id,e)$ and $i$ follows two actions $\va_i^*$ and $\va_i'$ where $i$ does not send more tuples containing $\id$ in $\va_i^*$ than in $\va_i'$ to any node,
then the probability of $i$ retrieving $e$ after following $\va_i^*$ is at least as high as after following $\va_i'$.
The complete proofs are in Appendix~\ref{app:corr}.

\begin{lemma}
\label{lemma:dissem-corr1}
Fix any history $h \in \hist$ from stage $t \in \mathbb{N}$, player $i \in \pset$, and identifier $\id \in \ids$.
The following properties hold:

\vspace{0.05in}
\textbf{D1.} If the source introduces $(\id,e)$ after the realisation of $h$,
then for all history $h^* \in \hist$ from stage $t$ succeeding $h$, player $j \in \pset$, and $(\id,e') \in \pend_j^{h^*}$, it holds $e' = g^h(e,j)$.

\vspace{0.05in}
\textbf{D2.} For every $t' > t$ and any two actions $\va_i^*,\va_i' \in \act_i$, the probability of $i$ receiving some tuple containing identifier $\id$ in stage $t'$
after $i$ follows $\va_i^*$ is the same as after $\va_i'$.

\vspace{0.05in}
\textbf{D3.} If $\id$ is introduced in stage $t$ after $h$ is realised, then for any two actions $\va_i^*,\va_i' \in \act_i$, the probability of $i$ receiving some tuple containing identifier $\id$ in stage $t$
after $i$ follows $\va_i^*$ is the same as after $i$ follows $\va_i'$.

\vspace{0.05in}
\textbf{D4.} Fix any two actions $\va_i^*,\va_i' \in \act_i$ that fulfil: i) $\stat_i^*(i) = \stat_i'(i)$ and $\miss_i^* = \miss_i'$; and ii) for any $j \in \pset$, $\id \in a_i^*(j)$ iff $\id \in a_i'(j)$.
If $\id$ is disseminated in stage $t$ before $h$ is realised, then $i$ receives some tuple containing $\id$ after following $\va_i^*$ iff the same holds after following $\va_i'$.

\vspace{0.05in}
\textbf{D5.} Let $e \in \evt^t$ be the event disseminated with identifier $\id$. Fix any two actions $\va_i^*,\va_i' \in \act_i^{h}$ 
that fulfil:  i) $\stat_i^*(i) = \stat_i'(i)$ and $\miss_i^* = \miss_i'$; and ii) for each $j \in \pset$, $\id \in a_i^*(j)$ only if $\id \in a_i'(j)$. 
For every outcomes $z^*,z' \in \term^t$ succeeding $\va_i^*$ and $\va_i'$, respectively, if $\id$ is introduced in stage $t$ prior to $h$ being realised
and $(\id,e) \in \rec_i(z')$, then $(\id,e) \in \rec_i(z^*)$.
\end{lemma}
\begin{proof}
(Sketch) \textbf{D1.} The source forwards $(\id,g^h(e,i))$ or no tuple containing $\id$ to each $i$. Inductively, if we suppose that D1 holds for round $r$,
then any $j$ sending $(\id,e')$ to $i$ in round $r+1$ has $e' = g^h(g^h(e,j),i)$, while $i$ adds $(\id,e'')$ to $\pend_i$, where $e'' = g^h(e,j)$. By Commutativity, $e'' = g^h(e,i)$.

\vspace{0.05in}
\textbf{D2.} We set $\stat_j(j) = \good$ and $\miss_j = \emptyset$ after first round of stage $t'$, and by Proposition~\ref{prop:facts}
these values are preserved throughout stage $t'$. Thus, every node $j$ forwards $\id$ immediately after its first reception to $l$
iff $l \in \prng_j^f(\seed_i,\id)$. Since the probability distribution over seeds is the same after both actions, then so is the probability of $i$ receiving $\id$.

\vspace{0.05in}
\textbf{D3.} This probability depends only on the value $\stat_j(j)$, $\miss_j$, and $\seed_j$ for every $j \in \pset \setminus \{i\}$: every node follows $\vsigma^f$ after $h$, during the dissemination of $\id$;
hence, $i$ receives $\id$ iff there exists a path of nodes terminating in $i$ such that every $j$ in that path forwards $\id$ to the next node $l$, which occurs when $\stat_j(j) = \good$, $\seq(\id) \notin \miss_j$,
and $l \in \prng_i^f(\seed_j,\id)$. By Proposition~\ref{prop:facts}, the values $\stat_j(j)$ and $\miss_j$ are never updated after $h$, regardless of what $i$ does in the present.
If $h$ follows the last monitoring round, then $\seed_j$ is already fixed since mediator is trusted, such that $i$ either receives $\id$ after $h$ or not,
independently of the current action. Otherwise, the probability distribution guiding the choice of $\seed_j$ is independent of the actions of $i$.

\vspace{0.05in}
\textbf{D4.} We use induction to show the result, starting from the base case that any node $j$ receives $\id$ for the first time immediately 
after $h$ with the same probability, regardless of the action of $i$. Also,  by Proposition~\ref{prop:facts}, $\stat_j(j)$ and $\miss_j$ are 
preserved after $h$ for every $j \in \pset \setminus \{i\}$, regardless of the current action of $i$. Thus, $j$
either forwards $\id$ after $h$ and after first receiving $\id$ to any node $l$ after $\va_i^*$ iff $j$ does the same after $\va_i'$.
This fact implies the induction step.

\vspace{0.05in}
\textbf{D5.} We prove two sub-properties. First, we show that, by not sending more identifiers in $\va_i^*$, the number of identifiers
being disseminated at each point in time after $\va_i'$ is a subset of that after $\va_i^*$, in a similar fashion to D4. Using this property,
we show that, if $(\id,e) \in \rec_i(z')$ and $i$ did not receive $(\id,e)$ in $h$, then it must hold $(\id,e) \in \rec_i(z^*)$. This follows from three facts.
First, $i$ can never forward $(\id,e')$ in $\va_i'$ before receiving it while being able retrieve $e$ from $e'$, since otherwise we would have by definition $(\id,e) \notin \rec_i(z')$.
Second, if $i$ forwards $(\id,e')$ in $\va_i'$ and this tuple loops back to $i$, then, despite nodes successively applying ciphers over $e'$, $i$ always receives $(\id,e')$
back. Thus, we avoid the scenario where $i$ receives $e' = (e)_{\key_i}$, forwards it causing some node to decipher $e$, and expects
$(\id,e)$ to return to $i$ without some node ciphering $e$ again. The third fact is the following. Suppose that $i$ believes that some tuple $(\id,e')$ is held
by some node $j$, and that this tuple is forwarded all the way to $i$ after $i$ follows $\va_i'$. Suppose also that $i$ can retrieve $e$ from the event resulting from the successive ciphering operations
applied over $e'$. Then, the exact same tuple reaches $i$ when $i$ follows $\va_i^*$. Here, the fact that nodes
deterministically forward the event received from the node with smallest identifier plays a key role. More precisely, by the first property,
the set of identifiers forwarded by any node after $\va_i^*$ is a subset of the same set after $\va_i'$. Thus, whenever some $j$
forwards $(\id,e')$ to $l$ after $\va_i'$ and $l$ prepares to forward $e'$ because $j$ has the smallest identifier, then $j$ also forwards $(\id,e')$
to $l$ and has the smallest identifier after $\va_i^*$ such that $l$ prepares to forward $e'$.
\end{proof}

Lemma~\ref{lemma:dissem-corr2} enumerates dissemination Properties~D6 and D7, which quantify the probability of node $i$ receiving each disseminated
event. D6 states that this probability is arbitrarily close to $q(f)$ for events disseminated in future stages, which is true by PRNG1. D7 states that
for any event with identifier $\id$ introduced in future rounds of the present stage this probability is at most $q(f) + \xi$ for an arbitrarily small $\xi$.
This follows from the fact that $i$ learns for each player $j \neq i$ at most the value $\prng_j^f(\seed_j,\id')$ for every $\id'$ previously introduced,
but $i$ does not learn these values for $\id$. Hence, any seed yielding any subset $\prng_j^f(\seed_j,\id)$ is possible and equally likely 
to have been chosen by the mediator, independently of the subsets generated for other events. By Consistency of $(\vsigma^f,\mu^f)$,
$i$ believes this to be true. By~PRNG1, it follows that the probability of any
node receiving $\id$ is at most arbitrarily close to $q(f)$, although it may be lower due to the existence of players that drop $\id$.

\begin{lemma}
\label{lemma:dissem-corr2}
Fix any player $i \in \pset$, stage $t \in \mathbb{N}$, private history $h_i \in \hist_i$ from stage $t$, and identifier $\id \in \ids$.
The following properties hold:

\vspace{0.05in}
\textbf{D6.} For every constant $\xi > 0$ and stage $t' > t$, there exists a function $\prng_i$ such that $i$ receives some tuple containing identifier $\id$ in stage $t'$ with probability $q'$ fulfilling $|q' -q(f)| < \xi$.

\vspace{0.05in}
\textbf{D7.} If $\id$ is introduced in stage $t$ after $h_i$ is observed, then for every constant $\xi > 0$ there exists a function $\prng_i$ such that $i$ believes
that will receive some tuple containing identifier $\id$ in stage $t$ with probability at most $q(f) + \xi$.
\end{lemma}

\subsection{Utility Analysis}
\label{sec:util}
In line with the One-deviation property, the goal of this section is to compute the difference between the expected utilities
$\eu{i}{\vsigma^f,\mu}{h_i,\va_i} - \eu{i}{\vsigma^f,\mu}{h_i,\va_i'}$ of player $i$ following two alternative actions $\va_i$ and $\va_i'$ after any private history $h_i$.
We divide the difference between expected utilities calculated from stage $t$ into three parts: \emph{long-term} utilities, referring to any stage $t' > t+1$;
\emph{medium-term} utilities relative to stage $t+1$; and \emph{short-term} utilities relative to stage $t$.
We will denote by $\gamma$ the cost of sending a tuple, i.e., $\gamma  =\alpha(c+ \lceil\log(\nevt)\rceil)$, where $c$ is the size of an event

\subsubsection{Long-term}
Lemma~\ref{lemma:future:diff} shows that the expected utility of any stage $t'>t+1$ is fixed, regardless of the present action in stage $t$.
The reason for this is that by our definition of strategy, regardless of what player $i$ does in stage $t$, he sends every requested message
in stage $t'$. Also, since he does not deviate in any future stage, he is never punished in stage $t'$, receiving all events with a fixed probability.

\begin{lemma}
\label{lemma:future:diff}
Fix any player $i \in \pset$, private history $h_i \in \hist_i$ from stage $t \in \mathbb{N}$, and any two actions $\va_i^*,\va_i' \in \act_i$.
For every $t' > t+1$, we have:
$$\eu{i}{\vsigma^f,\mu,t'}{h_i,\va_i^*} = \eu{i}{\vsigma^f,\mu,t'}{h_i,\va_i'}.$$
\end{lemma}
\begin{proof}
By construction, $i$ sets $\stat_i(i) = \good$ and $\miss_i=\emptyset$ at the beginning of stage $t'$, regardless of whether $i$ follows $\va_i^*$ or $\va_i'$.
By Proposition~\ref{prop:facts}, these variables are never updated during stage $t'$. Thus, $i$ always sends every requested monitoring message, 
implying that expected monitoring costs for stage $t'$ are the same whether $i$ follows $\va_i^*$ or $\va_i'$.
By~D2, $i$ receives a tuple with each identifier $\id \in \ids$ with the same probability after both actions.
Since $\stat_i(i) = \good$ and $\seq(\id) \notin \miss_i = \emptyset$ after any round when $i$ receives $\id$ for the first time,
$i$ forwards a tuple containing $\id$ to exactly $f$ nodes with the same probability after both actions. This implies that expected
costs of disseminating events are the same. Finally, by~M3 and M4, in any dissemination round $r$ of stage $t'$,
no node $j \in \pset \setminus \{i\}$ has an accusation of $i$. By D1, if $i$ receives $(\id,e')$ from $j$, then $e' = g^h(e,j)$ such that $i$ is capable of retrieving $e$ from $e'$.
Also, $i$ never forwards $\id$ before receiving a tuple with this identifier. Therefore, $(\id,e) \in \rec_i(z)$ for any $z \in \term^{t'}$ reached with positive probability such that $i$ 
a tuple containing $\id$ in $z$. By D2, $i$ receives a tuple containing $\id$ with the same probability after following $\va_i^*$ and $\va_i'$. Therefore,
the expected benefits are the same.
\end{proof}

\subsubsection{Medium-term}
Lemma~\ref{lemma:present-dissem:diff} analyses the difference between the expected utility of player $i$ following an action $\va_i^*$ prescribed
by the strategy for dissemination rounds and that of following an arbitrary action $\va_i'$. We consider three cases,
which cover all the scenarios of interest to show that $\vsigma^f$ is an equilibrium. Namely, we have $\stat_i^*(i) = \stat_i^{h_i}$
and $\miss_i^* = \miss_i^{h_i}$ for any observed $h_i \in \hist_i$ such that $i$ follows $\va_i^*$ after $h_i$.
This is because in $\va_i^*$ player $i$ forwards every event as specified by the strategy and does not send invalid messages.
By following $\va_i'$, $i$ may not update $\stat_i(i)$ or $\miss_i$ by following valid actions only and forwarding every event as expected,
may update $\stat_i(i)$ from $\good$ to $\bad$ after sending an invalid message, or add multiple sequences to $\miss_i$ when 
not forwarding some events as expected. However, $i$ may never update $\stat_i(i)$ from $\bad$ to $\good$ nor remove an element from $\miss_i$.

\begin{lemma}
\label{lemma:next:diff}
Fix any player $i \in \pset$, private history $h_i \in \hist$ from stage $t \in \mathbb{N}$, and any two actions $\va_i^*,\va_i' \in \act_i$.
The following holds regarding the value $\Delta^{t+1} = \eu{i}{\vsigma^f,t+1}{h,\va_i^*} - \eu{i}{\vsigma^f,t+1}{h,\va_i'}$:
\begin{enumerate}
  \item If $\stat_i^*(i) = \stat_i'(i) = \bad$, then $\Delta^{t+1} \geq 0$.
  \item If $\stat_i^*(i) = \good$ and $\stat_i'(i) = \bad$, then for any $\xi > 0$ there exists a function $\prng_i^f$ such that:
\begin{equation}
\label{eq:future-diff}
\Delta^{t+1} \geq (q(f) - \xi)\beta (1-\pr(|\miss_i^*|)).
\end{equation}
  \item If $\stat_i^*(i) = \stat_i'(i) = \good$, then for any $\xi > 0$ there exists a function $\prng_i^f$ such that:
\begin{equation}
\label{eq:future-smalldiff}
\Delta^{t+1} \geq (q(f) - \xi)\beta (\pr(|\miss_i'|)-\pr(|\miss_i^*|)).
\end{equation}
\end{enumerate}
\end{lemma}
\begin{proof}
Like in the previous lemma, $i$ initialises $\stat_i(i) = \good$ and $\miss_i = \emptyset$ in stage $t+1$, while never updating 
these variables during this stage as stated by Proposition~\ref{prop:facts}. Thus, expected monitoring costs
are the same whether $i$ follows $\va_i^*$ or $\va_i'$. By D2, $i$ receives a tuple for each $\id \in \ids$ with a fixed probability independent of the current
action, and forwards $\id$ to $f$ nodes. Thus, expected dissemination costs are also the same after both actions for any case.
Now, we analyse the expected benefits individually for each case.

\vspace{0.05in}
\textbf{Case 1.} This case is very similar to Lemma~\ref{lemma:future:diff}. By M2, $i$ is punished with probability $1$ after both actions.
By construction, if any $j \neq i$ sends $(\id,e')$ to $i$, then $(\id,e'') \in \pend_j$ where $e'' = g^h(e',i) = (e')_{\key_i}$.
By D1, $e'' = g^h(e,j)$, given that the source introduces $(\id,e)$. Thus, for every $(\id,e')$ sent to $i$, either we have $e' = (e)_{\key_i}$
or $((e)_{\key_i})_{\key_j}$ for some $j \neq i$. By non-disclosure, for any $z \in \term^{t+1}$ reached with positive probability
after any action $\va_i^*$ or $\va_i'$, $\rec_i(z) = \emptyset$, such that expected benefits are $0$.

\vspace{0.05in}
\textbf{Case 2.} By~M2, $i$ is punished in stage $t+1$ with probability $1$ after following $\va_i'$. As in Case 1, by~D1, the expected benefit is $0$.
On the other hand, by~M1, $i$ is punished in stage $t+1$ after following $\va_i^*$ with probability $\pr(|\miss_i^*|)$. When $i$ is not being punished,
by~D1 and~D6, $i$ obtains a benefit $\beta$ per disseminated event $e \in \evt^{t+1}$ with probability at least $q(f) - \xi$ for an arbitrarily small $\xi$:
$i$ receives a tuple $(\id,e')$ from some $j$ after any history $h$ with probability at least $q(f) - \xi$, such that $e' = g^h(e,j)$ and $i$ is able to retrieve
$e$ from $e'$; thus, for at least a fraction $q(f) - \xi$ of the histories $z \in \term^{t+1}$ reached with positive probability after $\va_i^*$ or $\va_i'$, 
we have $(\id,e) \in \rec_i(z)$. Therefore, the expected benefits after $i$ follows $\va_i^*$ are at least $(q(f) - \xi)\beta \nevt(1-\pr(|\miss_i^*|))/\nevt$, as we intended to prove.

\vspace{0.05in}
\textbf{Case 3.}
By~D1, $i$ retrieves each disseminated event iff $i$ is not being punished. By~M1, the probability of $i$ being punished after following $\va_i^*$ and $\va_i'$ is 
$x^* = \pr(|\miss_i^*|)$ and $x' = \pr(|\miss_i'|)$, respectively. By~D2 and D6, $i$ receives each event with some probability $q' \geq q(f) - \xi$ for an arbitrarily
small $\xi > 0$, both after $\va_i^*$ and $\va_i'$. This yields a difference between the expected benefits given as follows:
$$
\begin{array}{lll}
 & (q' \beta \nevt (1-x^*) - q' \beta \nevt (1 - x'))/\nevt &= \\
= & q' \beta (x' - x^*).
\end{array}
$$
This concludes the proof.
\end{proof}

\subsubsection{Short-term}
We need to differentiate between the cases where the alternative actions are followed in dissemination and monitoring rounds,
analysed in Lemmas~\ref{lemma:present-dissem:diff} and~\ref{lemma:present-monitor:diff}, respectively.
These lemmas use the following two auxiliary lemmas that analyse the expected utility relative to events disseminated in the present stage,
after the current round.

Lemma~\ref{lemma:present-future:bens} shows that the expected benefits of receiving events disseminated in future rounds are fixed, regardless of the present action.

\begin{lemma}
\label{lemma:present-future:bens}
Fix any player $i \in \pset$, history $h_i \in \hist_i$ from stage $t \in \mathbb{N}$, $\id \in \ids$ introduced after $h_i$ is observed, and any two actions $\va_i^*,\va_i' \in \act_i$. 
The expected benefits of receiving $e \in \evt^t$ introduced by the source with identifier $\id$ after $i$ follows $\va_i^*$ are the same as after $i$ follows $\va_i'$.
\end{lemma}
\begin{proof}
Fix any $h \in \hist$ such that $\mu^{\vsigma^f}(h|h_i) > 0$. Fix two action profiles $\va^* = (\va_i^*,\va_{-i}^*)$ and $\va' = (\va_i',\va_{-i}^*)$, where $\va_{-i}^*$ corresponds to an action profile 
prescribed by $\vsigma^*$ immediately after $h$, for every player other than $i$. Let $h^* = (h,\va^*)$ and $h' = (h,\va')$.
By~M4, $i$ is either punished in every dissemination round after $h^*$ and $h'$ by every node or not punished by any node.
If $i$ is punished, then by~D1, for any $z \in \term^t$ succeeding $h^*$ or $h'$, it holds $(\id,e) \notin \rec_i(z)$.
To see this suppose the opposite. By construction, there exists $j$ that sends $(\id,g^h(e',i))$ to $i$ in round $r$ such that $i$ is able to retrieve
$e$ from $g^h(e',i)$. However, since $j$ is following $\sigma_j^f$ in round $r$, it must hold $(\id,e') \in \pend_j$ after round $r-1$.
By D1, $e' = g^h(e,j)$. By Non-disclosure, $i$ cannot retrieve $e$ from $g^h(e',i)$, reaching a contradiction.
Therefore, the expected benefits of receiving events introduced in the future is $0$ when $i$ is being punished after $h$.

Now, consider that $i$ is not being punished after $h$. By~D3, the probability of $i$ receiving some tuple $(\id,e')$ when $\id$
is introduced after $h$ is the same after $h^*$ and $h'$. By~D1, when some $j \neq i$ sends $(\id,e')$ to $i$, it holds $e' = g^h(g^h(e,j),i) = g^h(e,j)$,
since $(\id,g^h(e',i)) \in \pend_j$ before $j$ sends $(\id,e')$ to $i$. Consequently, $i$ is able to retrieve $e$ from $e'$, obtaining a benefit $\beta$
with the same probability after $h^*$ and $h'$. This implies that expected benefits are the same in this case.
\end{proof}

Lemma~\ref{lemma:present-future:costs} quantifies the difference of the expected costs of following two alternative actions $\va_i^*$ and $\va_i'$, as a function
of the states resulting from following these actions. Namely, we first consider the case where $\va_i^*$ and $\va_i'$ result in $\stat_i^*(i) = \stat_i'(i) = \bad$ or 
$\stat_i^*(i) = \stat_i'(i) = \good$ and $\miss_i^* = \miss_i'$. This covers the scenarios where $i$ sends only valid messages prescribed by the strategy, except the exact content in each message may vary. In addition, when $\stat_i(i) = \bad$,
this also covers scenarios where $i$ sends invalid messages. Then, we analyse the case where, in $\va_i^*$, $i$ forwards every event as specified by the strategy
while, in $\va_i'$, $i$ fails to forward some event ($\stat_i^*(i) = \stat_i'(i) = \good$ and $\miss_i^* \subset \miss_i'$).
This is representative of any scenario where $\va_i^*$ is an action prescribed by the strategy, and $\va_i'$ does not contain invalid messages but 
also does not contain some tuple that $i$ was supposed to forward. Finally, we analyse the case when $i$ only sends valid messages in $\va_i^*$
and sends some invalid message in $\va_i'$, when $\stat_i(i)$ is initially $\good$.

\begin{lemma}
\label{lemma:present-future:costs}
Fix any stage $t \in \mathbb{N}$, player $i \in \pset$, history $h_i \in \hist_i$ from stage $t$, and any two actions $\va_i^*,\va_i' \in \act_i$. 
Regarding the difference between the expected costs of forwarding events introduced after $h_i$ is observed when $i$ follows $\va_i^*$ and when $i$ follows $\va_i'$, we have:
\begin{enumerate}
   \item If $\stat_i^*(i) = \stat_i'(i) = \bad$ or $\stat_i^*(i) = \stat_i'(i) = \good$ and $\miss_i^* = \miss_i'$, then the difference is $0$.
   \item If $\stat_i^*(i) = \stat_i'(i) = \good$ and $\miss_i^* \subset \miss_i'$, then for every $\xi > 0$ there exists a function $\prng_i^f$ such that 
  the difference is at least $- (q(f) + \xi)c \npseq \gamma f/\nevt$ for some constant $c>0$. 
  \item If $\stat_i^*(i) = \good$ and $\stat_i'(i) = \bad$, then for every $\xi > 0$ there exists a function $\prng_i^f$
  such that the difference is at least $-(q(f) + \xi) \gamma f$.
\end{enumerate}  
\end{lemma}
\begin{proof}
Fix any identifier $\id \in \ids$ introduced after $h_i$ is observed.

\vspace{0.05in}
\textbf{Case 1.} 
By Proposition~\ref{prop:facts} and the definition of $\vsigma^f$, $i$ forwards $\id$ exactly once after its first reception, and only if $\stat_i(i) = \good$. Thus, if $\stat_i^*(i) = \stat_i'(i) = \bad$, 
then the expected costs are $0$, whether $i$ follows $\va_i^*$ or $\va_i'$.
If $\stat_i^*(i) = \stat_i'(i) = \good$, then by~D3 $i$ receives a tuple containing $\id$ after $\va_i^*$ with the same probability as after $\va_i'$.
By Proposition~\ref{prop:facts} and the fact that $\miss_i^* = \miss_i'$, $i$ forwards $\id$ to exactly $f$ nodes
after $\va_i^*$ iff $i$ does so after $\va_i'$. Therefore, the expected costs are the same.

\vspace{0.05in}
\textbf{Case 2.} First, suppose that $\seq(\id) \in \miss_i' \setminus \miss_i^*$.
After following $\va_i'$, $i$ does not forward $\id$. After $\va_i^*$, by~D7, $i$ forwards $\id$ with probability at most $q(f) + \xi$
for an arbitrarily small $\xi$. Since there are at most $\npseq$ events per sequence and $|\miss_i'| - |\miss_i^*| \leq c$ for some constant $c$, this yields a difference of at least
$-(q(f) + \xi)c \npseq \gamma f/\nevt$. Now, suppose $\seq(\id) \in \miss_i^* \cap \miss_i'$. By~D3, $i$ receives $\id$ after $\va_i^*$ and $\va_i'$ with the same probability, in which
case $i$ forwards $\id$ to exactly $f$ nodes. Thus, the expected costs of forwarding $\id$ are the same.

\vspace{0.05in}
\textbf{Case 3.} $i$ does not forward $\id$ after $\va_i'$. By~D7, $i$ forwards $\id$ after $\va_i^*$ with probability at most $q(f) + \xi$.
Since there are at most $\nevt$ events introduced after $h_i$, this difference is at least $-(q(f) + \xi) \nevt \gamma f/\nevt$.
\end{proof}

Lemma~\ref{lemma:present-dissem:diff} analyses the difference between the expected utility of player $i$ following an action $\va_i^*$ prescribed
by the strategy for dissemination rounds, after the observation of $h_i$, and that of following an arbitrary action $\va_i'$. We consider three cases,
which cover all the scenarios relevant to the proof that $\vsigma^f$ is an equilibrium. By Proposition~\ref{prop:facts}, it holds $\stat_i^*(i) = \stat_i^{h_i}$ and $\miss_i^* = \miss_i^{h_i}$.
Case 1 covers the scenarios where $\stat_i^*(i) = \stat_i^{h_i} = \bad$, where we always have $\stat_i'(i) = \bad$, and covers the scenarios
where $\stat_i^*(i) = \good$ and $i$ sends only valid messages in $\va_i'$ and forwards every event as requested, implying that $\miss_i' = \miss_i^*$.
Case 2 considers the alternative scenario where $\stat_i^*(i) = \good$, $i$ only forwards valid messages in $\va_i'$ such that $\stat_i'(i) = \good$,
but $i$ fails to send some identifier $\id$ such that $\seq(\id) \in \miss_i' \setminus \miss_i^*$. Finally, case 3 analyses the scenario where $i$ sends
some invalid message in $\va_i'$ while $\stat_i^*(i) = \good$.

\begin{lemma}
\label{lemma:present-dissem:diff}
Fix any player $i \in \pset$, private history $h_i \in \hist$ from stage $t \in \mathbb{N}$ preceding a dissemination round, and any two actions $\va_i^*,\va_i' \in \act_i$
such that $\sigma_i^f(\va_i^*|h_i) > 0$. The following holds regarding the value $\Delta^t = \eu{i}{\vsigma^f,\mu,t}{h_i,\va_i^*} - \eu{i}{\vsigma^f,\mu,t}{h_i,\va_i'}$:
\begin{enumerate}
  \item If $\stat_i^*(i) = \stat_i'(i) = \bad$ or $\stat_i^*(i) = \stat_i'(i) = \good$ and $\miss_i^* = \miss_i'$, then $\Delta^t \geq 0$.
  \item If $\stat_i^*(i) = \stat_i'(i) = \good$ and $\miss_i^* \subset \miss_i'$, then for every $\xi > 0$ there exists a function $\prng_i^f$ and a constant $c> 0$ such that:
  $$\Delta^t \geq  - \dens(\beta + \gamma f)/\nevt - (q(f) + \xi)c \npseq \gamma f/\nevt.$$
  \item If $\stat_i^*(i) = \good$ and $\stat_i'(i) = \bad$, then for every $\xi > 0$ there exists a function $\prng_i^f$ such that:
  $$\Delta^t \geq - \dens(\beta + \gamma f)/\nevt - (q(f) + \xi) \gamma f.$$
\end{enumerate}
\end{lemma}
\begin{proof}
Fix $h \in \hist$ such that $\mu^{\vsigma}(h|h_i) > 0$, $\va^* = (\va_i^*,\va_{-i}^*)$, $\va' = (\va_i',\va_{-i}^*)$, $h^* = (h,\va^*)$, and $h' = (h,\va')$, such that $\va_{-i}^*$ is prescribed by $\vsigma^f$
for every player other than $i$. We analyse each case separately.

\vspace{0.05in}
\textbf{Case 1.} By Lemmas~\ref{lemma:present-future:bens} and~\ref{lemma:present-future:costs}, expected costs and benefits relative to events introduced after $h_i$ are the same.
Now, we analyse costs and benefits of events introduced before $h_i$.
Fix a tuple $(\id,e)$ introduced before $h$. If $\stat_i^*(i) = \bad$, then $i$ does not forward $\id$ in $\va_i^*$. Else, for any $j$, $\id \in a_i^*(j)$ iff $\seq(\id) \notin \miss_i^h$,
$j \in \prng_i^f(\seed_i,\id)$, and $i$ has received $\id$ for the first time in the previous round. In this case, 
it must hold $\id \in a_i^*(j)$ only if $\id \in a_i'(j)$. To see this, suppose that $\id \in a_i^*(j)$ and $\id \notin a_i'(j)$. 
This would imply $\seq(\id) \in \miss_i' \setminus \miss_i^*$, which would contradict case 1.
By~D5, $i$ obtains a benefit $\beta$ for receiving $(\id,e)$ after $h'$ only if $i$ obtains that benefit after $h^*$.
Thus, expected benefits relative to $(\id,e)$ after $i$ follows $\va_i^*$ are at least as high as those after $\va_i'$.

The expected costs of forwarding $\id$ depend on whether $i$ forwards $\id$ in $\va_i^*$ and $\va_i'$, or forwards $\id$ after these actions
are taken. We have seen that, for each $j \neq i$, $\id \in a_i^*(j)$ only if $\id \in a_i'(j)$.
Thus, costs of following $\va_i'$ are at least as high as those of following $\va_i^*$.
First, suppose that $\seq(\id) \in \miss_i^* = \miss_i'$ or $\stat_i^*(i) = \stat_i'(i) = \bad$.
By Proposition~\ref{prop:facts}, $i$ does not forward $\id$ after $h^*$ and $h'$. Therefore, expected costs of forwarding $\id$ after $h^*$ and $h'$ are both $0$.

Now, consider that $\seq(\id) \notin \miss_i^* = \miss_i'$ and $\stat_i'(i) = \stat_i^*(i) = \good$. By Proposition~\ref{prop:facts} and the definition of $\vsigma^f$, $i$ forwards $\id$ to $f$ nodes
after its first reception, both after $h^*$ and $h'$. It must hold for every $j \neq i$ that $\id \in a_i^*(j)$ iff $\id \in a_i'(j)$. To see this,
first suppose that $\id \in a_i^*(j)$ and $\id \notin a_i(j)'$ for some $j$. By definition of $\sigma_i^f$, $\seq(\id) \notin \miss_i^*$, which implies
that $i$ received $\id$ for the first time in the previous round and $j \in \prng_i^f(\seed_i,\id)$. But in this case we would have $\seq(\id) \in \miss_i' \setminus \miss_i^*$.
Now, suppose that $\id \notin a_i^*(j)$ and $\id \in a_i'(j)$. Again, $\seq(\id) \notin \miss_i^*$. Thus, either $j \notin \prng_i^f(\seed_i,\id)$ or $i$ did not receive $\id$
in the previous round. Either way, we would have $\seq(\id) \in \miss_i' \setminus \miss_i^*$. Both situations contradict the assumption that $\miss_i^* = \miss_i'$.
Therefore, we can apply D4 to conclude that $i$ receives $\id$ for the first time and forwards it after $h^*$ with the same probability as after $h'$.
Consequently, expected costs of forwarding $\id$ after $h^*$ and $h'$ are the same.

\vspace{0.05in}
\textbf{Cases 2 and 3.}  By Lemmas~\ref{lemma:present-future:bens} and~\ref{lemma:present-future:costs},
the difference between the expected utilities for receiving and disseminating events introduced after $h$ for
cases 2 and 3 is at least $-(q(f) + \xi)c \npseq \gamma f/\nevt$ and $-(q(f) + \xi) \gamma f$, respectively.
There are at most $\dens$ events introduced in $h$. In the worst scenario, when following $\va_i^*$,
$i$ does not retrieve any of those events while forwarding all of them, resulting in the minimum expected difference
of $-\dens (\beta + \gamma f)/\nevt$.
\end{proof}

In Lemma~\ref{lemma:present-dissem:diff}, we consider alternative actions $\va_i^*$ and $\va_i'$ followed in a monitoring round,
where $\va_i^*$ is prescribed by $\sigma_i^f$. Since we never update $\miss_i$ in a monitoring round, we only need to consider
the resulting values of $\stat_i^*(i)$ and $\stat_i'(i)$. By Proposition~\ref{prop:facts}, we always have $\stat_i^{h_i} = \stat_i^*(i)$,
where $h_i$ is the private history preceding $\va_i^*$. If $\stat_i^*(i) = \bad$, then $\stat_i'(i) = \bad$. Else, we have $\stat_i'(i) = \good$
iff $i$ sends all the required monitoring messages and only valid messages. These two scenarios are captured by case 1 of the lemma.
Case 2 encompasses the other scenario where $\stat_i^*(i) = \good$ and $i$ sends some invalid message or fails to send a requested message in $\va_i'$.

\begin{lemma}
\label{lemma:present-monitor:diff}
Fix any player $i \in \pset$, private history $h_i \in \hist$ from stage $t \in \mathbb{N}$ preceding a monitoring round, and any two actions $\va_i^*,\va_i' \in \act_i$.
The following holds regarding the value $\Delta^t = \eu{i}{\vsigma^f,\mu,t}{h_i,\va_i^*} - \eu{i}{\vsigma^f,\mu,t}{h_i,\va_i'}$:
\begin{enumerate}
  \item If $\stat_i^*(i) = \stat_i'(i)$, then $\Delta^t \geq 0$.
  \item If $\stat_i^*(i) = \good$ and $\stat_i'(i) = \bad$, then for every $\xi > 0$ there exists a function $\prng_i^f$ such that:
  $$\Delta^t \geq - n\alpha^* (1+ \nseq \pmon)/\nevt - (q(f) + \xi)\gamma f.$$
\end{enumerate}
\end{lemma}
\begin{proof}
Fix $h \in \hist$ such that $\mu^{\vsigma}(h|h_i) > 0$, $\va^* = (\va_i^*,\va_{-i}^*)$, $\va' = (\va_i',\va_{-i}^*)$, $h^* = (h,\va^*)$, and $h' = (h,\va')$,
where $\va_{-i}^*$ is prescribed by $\vsigma^f$ for every player other than $i$. We analyse each case separately.

\vspace{0.05in}
\textbf{Case 1.} First, consider that $\stat_i^*(i) = \stat_i'(i) = \bad$. By definition of $\sigma_i^f$ and Proposition~\ref{prop:facts}, $i$ does not send messages in $\va_i^*$ nor after $h^*$ and $h'$.
Thus, expected costs are at least as high when following $\va_i'$ as when following $\va_i'$. By Lemma~\ref{lemma:present-future:bens}, expected benefits are the same after $h^*$ and $h'$.
Now, suppose $\stat_i^*(i) = \stat_i'(i) = \good$. Since this is a monitoring round, it holds $\miss_i^* = \miss_i' = \emptyset$. By Lemmas~\ref{lemma:present-future:bens}
and~\ref{lemma:present-future:costs}, expected benefits and costs of receiving and forwarding events are the same after $h^*$ and $h'$. By Proposition~\ref{prop:facts},
$i$ forwards every requested monitoring message after $h^*$ and $h'$. Since mediator requests each report with fixed probability $p^*$, expected monitoring costs are the same.
Finally, $i$ sends at least as many messages in $\va_i'$ as in $\va_i^*$, implying that $\Delta^t \geq 0$.

\vspace{0.05in}
\textbf{Case 2.} Lemmas~\ref{lemma:present-future:bens} and~\ref{lemma:present-future:costs} imply that the contribution of the expected benefits
and costs of dissemination to $\Delta^t$ is at least $- (q(f) + \xi)\gamma f$. Regarding monitoring, $i$ sends in $\va_i^*$ all requested messages, which
cost at most $n \alpha^*$. By Proposition~\ref{prop:facts} and the definition of $\sigma_i^f$, $i$ does not send any monitoring message after $h'$, and
sends an expected number of $n \nseq \pmon$ messages after $h^*$: for each node $j$, $i$ sends reports to the mediator per sequence of identifiers with independent probability $\pmon$;
thus, $i$ sends reports relative to $j$ regarding an expected number of sequences $\nseq \pmon$. This adds at least the value $-n \alpha^*(1+ \nseq \pmon)/\nevt$ to $\Delta^t$.
\end{proof}

\subsection{Main Result}
\label{sec:res}
We now prove our main result. First, we establish sufficient conditions on the parameters $\npseq$, $\nseq$, and $\pmon$.
Then, we prove that $(\vsigma^f,\mu^f)$ is a SE. We conclude the section by proving our central claim.

\subsubsection{Conditions on Parameters}
The following conditions on $\npseq$, $\nseq$, and $\pmon$ are sufficient for our purposes:

\vspace{0.01in}
\textbf{C1.} $\npseq \log(\nevt) = o(\nevt)$.

\vspace{0.01in}
\textbf{C2.} $\pmon \log(\nevt) = o(1)$.

\vspace{0.01in}
\textbf{C3.} $\pmon \nevt = \Omega(\npseq)$.

\vspace{0.01in}
\textbf{C4.} For some constant $c> 0$, $\lim_{\nevt \to \infty}(1 - \pmon)^{\nseq} = c.$

C1 and C2 are required for keeping the expected monitoring costs sub-linear on $\nevt$. C3 and C4 guarantee
that the immediate gain of deviating does not out-weigh the future loss of being punished.
One definition that fulfils all conditions is $\npseq = \nseq = \sqrt{\nevt}$ and $\pmon = 1/\sqrt{\nevt}$.

\subsubsection{Equilibria Existence}
In Theorem~\ref{theorem:se}, we show that, for any fanout $f > 0$, $(\vsigma^f,\mu^f)$ is Sequentially Rational,
as long as $\delta$ and $\nevt$ are sufficiently large and the benefit $\beta$ per received event is sufficiently larger than
the cost $\gamma f$ of forwarding each event to $f$ nodes. Since $(\vsigma^f,\mu^f)$ is consistent by definition, it follows that $(\vsigma^f,\mu^f)$ is a SE.
In the proof, we use the fact that $\alpha^* = 4\alpha \npseq \log(\nevt)$, shown as follows. 
$\alpha^*$ comprises the cost of $i$ sending to the mediator $2 \npseq$ reports from $\se_i(j)$ and $\re_i(j)$.
Each report has a round and event identifier, corresponding to $2\log(\nevt)$ bits.

\begin{theorem}
\label{theorem:se}
For every $f>0$, there exist constants $c > 0$, $\bar{\delta} \in (0,1)$, and $\bar{\nevt} \in \mathbb{N}$, such that,
for every $\delta \in (\bar{\delta},1)$ and $\nevt > \bar{\nevt}$, if $\beta > c \gamma f$, then $(\vsigma^f,\mu^f)$ is Sequentially Rational.
\end{theorem}
\begin{proof}
We use the One-deviation Property to show the result. More precisely, we fix any player $i \in \pset$ and $h_i \in \hist_i$ from any stage $t$,
and we calculate, for any actions $\va_i^*,\va_i' \in \act_i$ such that $\sigma_i^f(\va_i^*|h_i) > 0$, the value:
$$\Delta = \eu{i}{\vsigma^f,\mu^f}{h_i,\va_i^*} -  \eu{i}{\vsigma^f,\mu^f}{h_i,\va_i'}.$$
The goal is to prove that $\Delta \geq 0$ for $\nevt$ and $\delta$ sufficiently large. We have $\Delta = \sum_{t' = t}^{\infty} \delta^{t'-t} \Delta^{t'}$, where:
$$\Delta^{t'} = \eu{i}{\vsigma^f,\mu^f,t'}{h_i,\va_i^*} - \eu{i}{\vsigma^f,\mu^f,t'}{h_i,\va_i'}.$$
By Lemma~\ref{lemma:future:diff}, it holds $\Delta^{t'} = 0$ for every $t' > t+1$, thus, we have $\Delta = \Delta^t + \delta \Delta^{t+1}$.
We need to compute $\Delta$ for the cases where $h_i$ precedes monitoring and dissemination rounds.

\vspace{0.1in}
\textbf{Monitoring Round.} When $\stat_i'(i) = \stat_i^*(i)$, $\Delta \geq 0$ by 1 of Lemmas~\ref{lemma:next:diff}
and~\ref{lemma:present-monitor:diff}. Else, by 2 of the same lemmas, we have for an arbitrarily small $\xi > 0$:
$$\Delta \geq - n\alpha^* (1+ \nseq \pmon)/\nevt  - (q(f) + \xi) \gamma f + \delta (q(f) - \xi) \beta.$$
Since $\nseq \npseq = \nevt$, by~C1 and~C2:
$$
\begin{array}{ll}
n\alpha^* (1+ \nseq \pmon)/\nevt  & = 4n \alpha \npseq \log(\nevt) (1 + \nseq \pmon)/\nevt\\
 & = O(\pmon\log(\nevt) + \npseq \log(\nevt)/\nevt) = o(1).
\end{array} 
$$
Therefore, for $\nevt$ sufficiently large, there exists $\epsilon$ arbitrarily small such that:
$$\Delta \geq  q(f)(\delta \beta -  \gamma f) - \epsilon.$$
If $\beta > \gamma f$ and $\delta$ is sufficiently large, then $\Delta \geq 0$.

\vspace{0.1in}
\textbf{Dissemination Round.} We need to consider the three cases where: 1) $\stat_i^*(i) = \stat_i'(i) = \good$ and $\miss_i^* = \miss_i'$ or $\stat_i^*(i) = \stat_i'(i) = \bad$;
2) $\stat_i'(i) = \stat_i^*(i) = \good$ and $\miss_i^* \subset \miss_i'$; and 3) $\stat_i^*(i) = \good$ and $\stat_i'(i) = \bad$.

\vspace{0.05in}
\textbf{Case 1.} By 1) of Lemmas~\ref{lemma:next:diff} and~\ref{lemma:present-dissem:diff}, $\Delta \geq 0$.

\vspace{0.05in}
\textbf{Case 2.} By 2) of Lemmas~\ref{lemma:next:diff} and~\ref{lemma:present-dissem:diff}, we have for some constant $c$ and an arbitrarily small $\xi > 0$:
$$
\begin{array}{ll}
\Delta \geq &- o(\nevt)/\nevt - (q(f) + \xi)c (\npseq/\nevt) \gamma f + \\
   & +\delta (q(f) - \xi) \beta (\pr(x') - \pr(x^*)),
\end{array}
$$
where $x^* = |\miss_i^*|$ and $x' = |\miss_i'|$. Since $x' - x^* \geq 1$:
$$
\begin{array}{lll}
& \pr(x') - \pr(x^*) &\geq\\
\geq&(1-(1-\pmon)^{x^*+1}) - (1-(1-\pmon)^{x^*}) &= \\
=&\pmon(1 - \pmon)^{x^*} \geq \pmon(1-\pmon)^{\nseq}.
\end{array}
$$ 
If $\nevt$ is sufficiently large, then for an arbitrarily small $\epsilon$:
$$\Delta \geq q(f)(\delta \beta \pmon(1-\pmon)^{\nseq} - c(\npseq/\nevt)\gamma f) - \epsilon.$$
If we divide both sides by $k = \pmon(1-\pmon)^{\nseq}$, then by~C3 and C4, there exists $c' > 0$ such that for an arbitrarily small $\epsilon > 0$:
$$\Delta/k \geq q(f)(\delta \beta - c'\gamma f) - \epsilon.$$
If $\beta > c' \gamma f$, then for $\delta$ sufficiently close to $1$ and $\epsilon$ sufficiently small we have $\Delta \geq 0$.

\vspace{0.05in}
\textbf{Case 3.} By 3) of Lemmas~\ref{lemma:next:diff} and~\ref{lemma:present-dissem:diff}, for an arbitrarily small $\xi > 0$:
$$
\Delta \geq - o(\nevt)/\nevt - (q(f) + \xi)\gamma f + \delta (q(f) - \xi) \beta (1 - \pr(x^*)),
$$
where $x^* = |\miss_i^*|$. Since $1 - \pr(x^*) \geq (1-\pmon)^{\nseq}$, as in case 2), there exists a constant $c>0$ such that
for an arbitrarily small $\epsilon$:
$$\Delta \geq q(f)(\delta \beta - c\gamma f) - \epsilon.$$
If $\beta > c\gamma f$, for $\delta$ sufficiently close to $1$, it holds $\Delta \geq 0$.
\end{proof}

\subsubsection{Central Claim}
We conclude with the proof of Theorem~\ref{theorem:folk}. The result follows trivially for $f=q(f)= 0$. Thus, fix any $f  > 0$. 
We have already defined a monitoring mechanism and a pair $(\vsigma^f,\mu^f)$ that is a SE of the induced game, as long as $\delta$ and $\nevt$ are sufficiently large,
and $\beta = \Omega(\gamma f)$. By M1, M3, and~D6, the expected utility of forwarding and receiving events  per stage is $\tilde{u}(f)$
fulfilling $|\tilde{u}(f) - \bar{u}(f)| < \epsilon$ for an arbitrarily small $\epsilon$. This yields an average utility of:
$$
(1-\delta)\sum_{t = 0}^{\infty} \delta^t (-A/\nevt + \tilde{u}(f)) =- A/\nevt + \tilde{u}(f),
$$
where $A \leq n\alpha^*(1+\pmon \nseq)$.
By~C1 and~C2, we have $A = o(\nevt)$. For $\nevt$ sufficiently large, $A/\nevt$ is arbitrarily small,
and the average utility is arbitrarily close to $\bar{u}(f)$.

\section{Discussion}
\label{sec:discuss}
The main step towards improving our result lies in distributing the role of the mediator, possibly having a different
mediator per node. This goal poses two main challenges: 1) how to ensure that every node agrees on whether to punish any $i$
and on the key $\key_i$; and 2) how to distribute the seeds. The first challenge can be addressed by having nodes reaching an agreement 
regarding the set of punishments and keys for each stage, in a similar fashion to~\cite{Dolev:10}. The second challenge can be addressed using techniques similar
to those from~\cite{Abraham:13}, by having nodes committing to random numbers without revealing each seed $\seed_i$ to any node other than $i$.

\section*{Acknowledgements}
This work was partially supported by Funda\c{c}\~{a}o para a Ci\^{e}ncia e Tecnologia (FCT) 
via the INESC-ID multi-annual funding through the PIDDAC Program fund grant, under 
project PEst-OE/ EEI/ LA0021/ 2013, and via the project PEPITA (PTDC/EEI-SCR/2776/2012).

\bibliographystyle{abbrv}

\bibliography{biblong}

\appendix
\section{Correctness Proofs}
\label{app:corr}
\vspace{0.1in}
{\large \textbf{Proof of Lemma~\ref{lemma:dissem-corr1}}}

\textbf{D1.} We use induction to show that the property holds for every history $h'=h$ or succeeding $h$, and preceding $h^*$. For the base case, we have $h' = h$. Since $\id$ has not been disseminated
prior to $h$, no node $j$ has $\id$ in $\pend_j$. For the induction step, suppose that D1) holds for some history $h'$, and fix any $j$. If $(\id,e') \in \pend_j^{h'}$,
then the result follows by the hypothesis. If $j$ is the source and disseminates $\id$ when $h'$ is observed, then only $j$ adds $(\id,e)$ to $\pend_j$ after $h'$.
Else, consider that $j$ receives $(\id,e')$ from $l \in \pset$. By the hypothesis, we have $e' = g^h(g^h(e,l),j)$, and $j$ adds $(\id,e'')$ to $\pend_j$ 
with $e'' = g^h(e',l)$. By Commutativity, $e'' = g^h(e,j)$.

\vspace{0.05in}
\textbf{D2.}
Let $\va^* = (\va_i^*,\va_{-i})$ and $\va' = (\va_i',\va_{-i})$ be the action profiles followed after $h$ when players different from $i$ abide to $\vsigma^f$.
Fix any two histories $h^*$ and $h'$ from stage $t'$, immediately preceding the last monitoring round of stage $t'$, and 
succeeding $(h,\va^*)$ and $(h,\va')$, respectively.
By the fact that we initialise $\stat_j(j) = \good$ and $\miss_j = \emptyset$ in stage $t'$, and by Proposition~\ref{prop:facts},
it holds $\stat_j^{h^*}(j) = \stat_j^{h'}(j) = \good$ and $\miss_j^{h^*} = \miss_j^{h'} = \emptyset$ for every $j \in \pset$.
Thus, $j$ abides to $\vsigma^f$ throughout stage $t'$ by forwarding a tuple containing $\id$ once after its first reception, to the $f$
nodes from $\prng_j^f(\seed_j,\id)$. Thus, if the selection of seeds after $h^*$ is the same as after $h'$, then $i$ receives a tuple containing $\id$ after $h^*$ iff $i$ receives $\id$ after $h'$.
Since the probability distribution of the selection of seeds is the same after $h^*$ and $h'$, the result follows.

\vspace{0.05in}
\textbf{D3.} Let $\va^* = (\va_i^*,\va_{-i})$ and $\va' = (\va_i',\va_{-i})$ be the action profiles followed after $h$ when players different from $i$ abide to $\vsigma^f$.
First, suppose that $h$ precedes a dissemination round. Since the mediator is trusted, regardless of the action followed by $i$, the value of $\seed_j$ is fixed for every $j \in \pset$.
If $j \neq i$, then by Proposition~\ref{prop:facts} whether $i$ follows $\va_i^*$ or $\va_i'$ the values $\stat_j(j)$ and $\miss_j$ in any dissemination round are equal to $\stat_j^h(j)$ and $\miss_j^h$, respectively.
Therefore, the behaviour of $j$ after $i$ follows $\va_i^*$ is the same as after $i$ follows $\va_i'$: $j$ either forwards $\id$ after its reception according to $\prng_j^f(\seed_j,\id)$,
or $j$ does not forward $\id$ to any node. Thus, $i$ receives a tuple containing $\id$ after following $\va_i^*$ iff $i$ receives $\id$ after following $\va_i'$. 
Now, suppose that $h$ precedes a monitoring round and let $h^*$ and $h'$ be any histories immediately preceding the last monitoring round
and succeeding $(h,\va^*)$ and $(h,\va')$, respectively. As in Property~D2, for the same selections of seeds after $h^*$ and $h'$, 
$i$ receives $\id$ after $h^*$ iff $i$ receives $\id$ after $h'$. Since the probability distribution over
the seeds is the same, $i$ receives $\id$ with the same probability.

\vspace{0.05in}
\textbf{D4.}  Let $\va^* = (\va_i^*,\va_{-i})$ and $\va' = (\va_i',\va_{-i})$ be the action profiles followed after $h$ when players different from $i$ abide to $\vsigma^f$.
$i$ receives $\id$ in $(h,\va^*)$ iff the same holds in $(h',\va')$. Now, we show using induction that, for every $j \neq i$ and histories $h^*$ and $h'$ 
succeeding or equal to $(h,\va^*)$ and $(h,\va')$ respectively: 1) $j$ receives $\id$ in $h^*$ iff $j$ receives $\id$ in $h'$; and 2) $j$ forwards $\id$ immediately after $h^*$ iff $j$ forwards $\id$ immediately after $h'$.
This proves that $i$ receives $\id$ after $\va^*$ iff $i$ receives $\id$ after $\va'$.

\vspace{0.01in}
\textbf{Base case.} 1) holds if $j$ receives $\id$ in $h$. Otherwise, $j$ receives $\id$ in $\va^*$ iff $j$ receives $\id$ in $\va'$, proving 1). By construction,
$\stat_j^{h^*}(j) = \stat_j^{h'}$ and $\miss_j^{h^*} = \miss_j^{h'}$. This means that $j$ has $\stat_j^{h^*}(j) = \bad$
or $\seq(\id) \in \miss_i^{h^*}$ iff $\stat_j^{h'}(j) = \bad$ or $\seq(\id) \in \miss_i^{h'}$. This implies that $j$ forwards $\id$ immediately after $h^*$
iff the same holds after $h'$, proving 2). 

\vspace{0.01in}
\textbf{Induction step.} Consider that the hypothesis holds for any $h^*$ and $h'$. It holds that $j$ receives $\id$ in $h^*$
iff $j$ receives $\id$ in $h'$, implying 1. In this case, $j$ does not forward $\id$ one round after $h^*$ and $h'$, which implies 2). 
Now, by 2) of the hypothesis, $j$ first receives $\id$ from some $l$ after $h^*$ iff $j$ first receives $\id$ from $l$ after $h'$, proving 1). By Proposition~\ref{prop:facts}, 
$\stat_j^{h^*}(j) = \stat_j^{h'}(j)$ and $\miss_j^{h^*} = \miss_j^{h'}$, proving 2).

\vspace{0.05in}
\textbf{D5.} Fix any $z^*$ and $z'$. Let $\va^* = (\va_i^*,\va_{-i})$ and $\va' = (\va_i',\va_{-i})$ be the action profiles immediately following $h$ in $z^*$ and $z'$, respectively.
$i$ receives $(\id,e')$ in $(h,\va^*)$ iff $i$ receives $(\id,e')$ in $(h,\va')$, whether $i$ can retrieve $e$ from $e'$. Thus, if $(\id,e) \in \rec_i(z')$
because $i$ received $(\id,e')$ in $(h,\va')$, then it also holds $(\id,e) \in \rec_i(z^*)$. 

Now, suppose that for all $(\id,e')$ received in $(h,\va')$ and $(h,\va^*)$ it holds that $i$ cannot retrieve $e$ from $e'$.
Suppose also that $(\id,e) \in \rec_i(z')$. Here, some $j \neq i$ sends $(\id,e'')$ to $i$ after $(h,\va')$ such that $i$ can retrieve $e$ from $e''$ or $g^h(e'',j)$. This occurs after some history $\bar{h}$,
such that $(\id,e') \in \pend_j^{\bar{h}}$ with $e' = g^h(e'',i)$, and $i$ is able to retrieve $e$ from $g^h(g^h(e',j),i)$. In this case,
we say that $i$ can retrieve $e$ from $e'$ in $j$. 

We prove two sub-properties. First, we show using induction that, 
for any histories $h'$ and $h^*$ succeeding or equal to $(h,\va')'$ and $(h,\va^*)$ respectively, any $j \neq i$ forwards $\id$ immediately after $h^*$ only if $j$ forwards $\id$ immediately after or in $h'$. 
Then, we use this fact to prove through induction that, for any histories $h'$ and $h^*$ succeeding $(h,\va')$ and $(h,\va^*)$ respectively, and any $j \neq i$, if $i$ can retrieve $e$ from $e'$ in $j$ then:
1) $(\id,e') \in \pend_j^{h'}$ only if $(\id,e') \in \pend_j^{h^*}$; and 2) $j$ forwards $(\id,e')$ one round after $h'$ only if $j$ forwards $(\id,e')$ one round after $h^*$. Given that
these properties hold, $i$ receives $(\id,e')$ from any $j$ after $(h,\va')$ only if $i$ receives $(\id,e')$ from $j$ after $(h,\va')$. Thus,
if $(\id,e) \in \rec-i(z')$, then $(\id,e) \in \rec_i(z^*)$.

\vspace{0.02in}
\textbf{Property 1.}
\emph{Base case.} It is true that if $j$ has received $\id$ in $h$, then $j$ never forwards $\id$ after $(h,\va^*)$ and $(h,\va')$.
In addition, $j$ receives $\id$ in $\va^*$ for the first time only if $j$ receives $\id$ in $\va'$ for the first time. That is, $j$ receives $\id$ from any $l\neq i$ in $\va^*$
iff the same holds in $\va'$, whereas $j$ never receives $\id$ in $\va^*$ from $i$ while not receiving it from $i$ in $\va'$. 
By the preservation of $\stat_j(j)$ and $\miss_j$, it holds that $j$ forwards $\id$ immediately after $\va^*$ only if $j$ forwards $\id$ after $\va'$. 

\emph{Induction step.} Suppose the hypothesis holds for two histories $h^*$ and $h'$. Let $h^1 = (h^*,\va^1) \in z^*$
and $(h',\va^2) \in z'$ be the histories immediately succeeding $h^*$ and $h'$, respectively.
Fix any $j \neq i$. If $j$ already received and forwarded $\id$ in $h^1$, then by the hypothesis
$j$ already forwarded $\id$ in $h^2$, and does not forward $\id$ after $h^2$.
Now, consider that $j$ receives $\id$ for the first time immediately after $h^*$, and forwards $\id$ immediately after $h^1$. By construction,
it must hold that $\stat_j^{h^1}(j) = \good$ and $\seq(\id) \notin \miss_j^{h^1}$.
By Proposition~\ref{prop:facts}, $\stat_j^{h^1}(j) = \stat_j^{h^2}(j)$ and $\miss_j^{h^1} = \miss_j^{h^2}$.
By the hypothesis, either $j$ forwarded $\id$ to $j$ in $h^2$ or $j$ receives $\id$ immediately after $h'$.
In the former case, the result follows. In the latter case, since we also have $\stat_j^{h^2}(j) = \good$ and $\seq(\id) \notin \miss_j^{h^2}$,
$j$ forwards $\id$ immediately after $h^2$. This proves the result.

\vspace{0.02in}
\textbf{Property 2.} \emph{Base case.} Suppose $j$ receives $(\id,e'')$ from $l$ with $e'' = g^h(e',l)$, such that $l$ is the node
with the smallest identifier among the nodes sending $\id$ to $i$ in $\va'$, and adds $(\id,e')$ to $\pend_j^{h'}$.
We cannot have $l = i$, or else $i$ sends $(\id,e'')$ to $j$ before receiving this tuple and such that $i$ can retrieve $e$ from $e''$,
implying that $(\id,e) \notin \rec_i(z')$. Therefore, $l$ also sends $(\id,e'')$ to $j$ in $\va^*$. The set of nodes sending $(\id,e'')$ to $j$ in $\va^*$ is a subset of those
that send $\id$ in $\va'$. Thus, $j$ also receives $(\id,e'')$ from $l$, $l$ also has the smallest identifier, and $j$ adds $(\id,e')$ to $\pend_j^{h^*}$,
proving the base case.

\emph{Induction step.} Suppose the hypothesis holds for $h'$ and $h^*$. Let $h^1 = (h^*,\va^1) \in z^*$
and $(h',\va^2) \in z'$ be the histories immediately succeeding $h^*$ and $h'$, respectively.
Fix any $j \neq i$ and $e'$ such that $i$ can retrieve $e$ from $e'$ in $j$.
By the hypothesis, if $(\id,e') \in \pend_j^{h'}$, then $(\id,e') \in \pend_j^{h^*}$; if $j$ forwards $\id$ immediately after $h'$,
then $j$ forwards $\id$ immediately after $h^*$. Either way, $j$ does not forward $\id$ after $h^1$ and $h^2$, proving the result.
Now, suppose that $(\id,e') \notin \pend_j^{h'}$ and that $j$ receives $\id$ for the first time immediately 
after $h'$ from a set of nodes $S'$. Let $l \in S'$ be the node with the smallest identifier, which may be $i$.

If $l$ sends $(\id,e'')$ with $e'' \neq g^h(e',l)$, then $j$ never adds $(\id,e')$ to $\pend_j$ and never forwards this tuple after $h'$. Otherwise, $(\id,e') \in \pend_j^{h^2}$
and $(\id,e^1) \in \pend_l^{h'}$, where $e^1 = g^h(e'',j)$. Since $i$ can retrieve $e$ from $e'$ in $j$, $i$ can retrieve $e$ from $e^1$ in $l$:
$i$ can retrieve $e$ from $g^h(e^2,i)$ where $e^2 = g^h(e',j)$; by Commutativity, we have $e^1 = g^h(g^h(e',l),j) = g^h(e^2,l)$; thus, $i$ can retrieve $e$ from $g^h(g^h(e^1,l),i) = g^h(e^2,i)$.
Therefore, we can apply 2 of the hypothesis to conclude that $(\id,e^2) \in \pend_l^{h^*}$, and $l$ also sends $(\id,e'')$ immediately after $h^*$ to $j$.
Let $S^*$ be the set of nodes sending $\id$ to $j$ immediately after $h^*$. It must hold that $j$ receives $\id$ for the first time after $h^*$, or else we would reach a contradiction
with Property 1): if $j$ had received $\id$ in $h^*$, then he would also have received $\id$ in $h'$, contradicting the supposition that $j$ receives $\id$ for the first time after $h'$.
It thus follows from Property~1) that, if any node $x$ sends $\id$ to $i$ immediately after $h^*$, then $x$ sends $\id$ immediately after $h'$, i.e., $S^* \subseteq S'$. Since $l \in S^*$, 
$l$ is the node with the smallest identifier in $S^*$. Thus, $j$ adds $(\id,e')$ to $\pend_j^{h^1}$, proving 1. In this case, by Proposition~\ref{prop:facts},
$\miss_j^{h^*} = \miss_j^{h'}$ and $\stat_j^{h^*}(j) = \stat_j^{h'}(j)$, and $j$ forwards $\id$ immediately after $h^1$ iff $j$ forwards $\id$ after $h^2$.
This proves the result.

\vspace{0.1in}
{\large \textbf{Proof of Lemma~\ref{lemma:dissem-corr2}}}
Properties~D6) and D7) follow from Lemmas~\ref{lemma:reliable-future-equiv} and~\ref{lemma:reliable-present-equiv}, respectively,
proven in Appendix~\ref{sec:dissem-model}.

\section{Pseudo Random Number Generator}
\label{sec:prng}
We consider a pseudo-random number generator function defined as follows. The definition used in~\cite{Abraham:13}
considers a pseudo-random number generator function $G$ defined as a sequence of $m(k)> k$ bits generated from a seed $s \in \{0,1\}^k$ for some security parameter
$k$. $G$ is said to be a pseudo-random number generator if no probabilistic time machine can distinguish between the outcome of $G$ and truly random sequences
of bits. We formalise this intuition in a way that is useful for our purposes as follows. For all $\xi > 0$ and polynomial $m$,
there exists $k$ such that for every sequence of bits $\vec{b} \in \{0,1\}^{m(k)}$:
\begin{equation}
\label{eq:prng}
\left| \frac{1}{2^{m(k)}} - \frac{|S|}{2^{k}}\right| < \xi,
\end{equation}
where $S$ is the largest set of seeds such that for every $s \in S$ we have $\vec{b} = G(s)$.

For each $i \in \pset$, we define $\prng_i^f$ from $G$ as follows. Let $x = \binom{n-1}{f}$ be the number of different subsets of nodes to whom $i$ may send any event.
We let $m(k) = \nevt \cdot \phi(k) \cdot \lceil\log(x)\rceil$ for some polynomial $\phi(k)$. Any generated stream represents a sequence of $\nevt$ numbers. Given the $\id$-th number $b$, 
$\prng_i^f(\seed_i,\id)$ returns the $y$-th subset, where $y \equiv b \mbox{ (mod $x$)}$. By the definition of $G$, the probability of selecting a particular
sequence of numbers is arbitrarily close to $1/x^{\nevt}$, shown as follows. By using the modulo operation, we ensure that there exist
two integers $(d,r)$ such that: i) $r \in \{0 \ldots x-1\}$; ii) $m(k) = x d + r$; iii) $r$ subsets of each $\id$ are mapped into $d+1$ numbers; and iv) $x-r$ subsets are mapped into $d$ numbers.
The probability of each sequence of subsets being selected is arbitrarily close to selecting for each $\id$ a number that yields the corresponding subset. Since for each $\id$
and subset, there are $d$ or $d+1$ different numbers yielding that subset, by~\ref{eq:prng}, the probability of selecting each subset is either arbitrarily close to $d/m(k)$ or $(d+1)/m(k)$. Thus, the probability
of selecting the sequence is arbitrarily close to a value lying in the interval:
$$\left(\left(\frac{d}{xd + r}\right)^\nevt,\left(\frac{d+1}{xd + r}\right)^{\nevt}\right).$$ 
By taking the limit $\phi(k) \to \infty$, we also have $d \to \infty$, and the probability of each sequence of subsets being selected converges to $1/x^\nevt$.
Thus, $\prng_i^f$ fulfils Property PRNG1 as shown by Proposition~\ref{prop:prng1}. Here, we can adjust the constant $\xi$ by adjusting $k$ and $\phi(k)$.

\begin{proposition}
\label{prop:prng1}
Fix any node $i \in \pset$ and sequence of subsets $\vec{S} = (S_{\id})_{\id \in \ids}$ with $S_{\id} \subseteq \pset \setminus \{i\}$ and $|S_{\id}| = f$ for every $\id \in \ids$.
For any constant $\xi > 0$, there exist $k$ and $\phi(k)$ such that for every $\id \in \ids$:
$$\left|\pra{\vec{S}}{\prng_i^f}{\vec{S}_{-\id}} - \frac{1}{\binom{n-1}{f}}\right| \leq \xi.$$
\end{proposition}
\begin{proof}
Fix any $\xi > 0$. By the definition of $\prng_i^f$, we can define $k$ and $\phi(k)$ such that:
\begin{equation}
\label{eq:prng1}
|1/x^{\nevt} - \pr^{\prng_i^f}(\vec{S})| < \xi,
\end{equation}
where $\pr^{\prng_i^f}(\vec{S}) = \#S/2^k$ is the probability of selecting a seed such that $\prng_i^f$ yields $\vec{S}$.
We can write:
$$\pra{\vec{S}}{\prng_i^f}{\vec{S}_{-\id}} = a/b = (a/2^k)/(b/2^k) = \pr^{\prng_i^f}(\vec{S})/\pr^{\prng_i^f}(\vec{S}_{-\id}),$$ 
where $a$ is the number of seeds that yield $\vec{S}$, and $b$ is the number of seeds that yield $\vec{S}_{-\id}$,
independently of $S_\id$. In particular, by~\ref{eq:prng}, we can define $\prng_i^f$ such that:
$$
\begin{array}{lll}
& \pr^{\prng_i^f}(\vec{S}_{-\id}) &= \\
=& \sum_{S_\id' \subseteq \pset \setminus \{i\}: |S_\id'| = f} \pr^{\prng_i^f}(S_\id',\vec{S}_{-\id})& >\\
> & \sum_{S_\id' \subseteq \pset \setminus \{i\}: |S_\id'| = f} (1/x^\nevt - \xi) = x*(1/x^\nevt - \xi).
\end{array}
$$
Similarly:
$$\pr^{\prng_i^f}(\vec{S}_{-\id}) < x*(1/x^\nevt + \xi).$$
Therefore, we can write:
$$
\begin{array}{lll}
& \pr^{\prng_i^f}(\vec{S})/\pr^{\prng_i^f}(\vec{S}_{-\id}) &>\\
>& (1/x^\nevt - \xi)/(x(1/x^\nevt+ \xi)) & = \\
=& 1/(x(1 + x^\nevt\xi)) - \xi x^{\nevt}/(x(1+ x^\nevt \xi)) & =\\
=& 1/x - 2\xi x^{\nevt}/(x(1 + x^\nevt\xi)).
\end{array}
$$
Thus, we can define $\prng_i^f$ such that for an arbitrarily small $\xi'>0$ it holds
$\pr^{\prng_i^f}(\vec{S})/\pr^{\prng_i^f}(\vec{S}_{-\id}) > 1/x - \xi'$. Using an identical reasoning, we can also 
define $k$ and $\phi(k)$ such that $\pr^{\prng_i^f}(\vec{S})/\pr^{\prng_i^f}(\vec{S}_{-\id}) < 1/x + \xi'$. This
shows that $|\pra{\vec{S}}{\prng_i^f}{\vec{S}_{-\id}} - 1/x| < \xi'$ for an arbitrarily small $\xi'$, as we intended to prove.
\end{proof}

\section{Dissemination Model}
\label{sec:dissem-model}
For any fanout $f$, the reliability of dissemination $q(f)$ is the probability of any node $i$ receiving an event, and is 
determined using a model similar to the SIR model\,\cite{Abbey:52}. More precisely, the source introduces $e \in \evt^t$ with identifier $\id$ in 
round $r^* = \id + \nrndm-1$ of stage $t$. In every round $r$, $Q^{\id,f}(S,I,R,r)$ is the probability that nodes from the set $S$ have received and 
forwarded a message containing $\id$ by round $r$, $I$ contains the nodes that received a message containing $\id$ and will forward it in the next round, 
and $R$ contains the nodes that have not received $\id$, such that $(S,I,R)$ forms a partition of $\pset$.
This probability is defined recursively as $Q^{\id,f}(\emptyset,\emptyset,\pset,r) = 1$ for $r<r^*$,
$Q^{\id,f}(\emptyset,\{s\},\pset \setminus \{s\},r^*) = 1$, where $s$ is the source, and for any $r > r^*$, partition $(S,I,R)$, and $I' \subseteq R$:
$$
\begin{array}{ll}
  &Q^{\id,f}(S,I,R,r+1) =\\
 =&\sum_{I' \subseteq S} Q^{\id,f}(S \setminus I',I',R \cup I,r) P^{\id,f}(I | I', R),
 \end{array}
 $$
where $P^{\id,f}(I | I',R)$ is the probability of exactly all nodes in $I$ among $R \cup I$ receiving a dissemination message containing $\id$ from some node in $I$,
given that each node forwards $\id$ to a subset of $f$ neighbours chosen uniformly at random. 
The dissemination ends in round $r = \nrndm + \id + \dens$. Then, $q(f)$ is the probability of any node $i$ receiving $\id$ during the present stage and is defined as:
\begin{equation}
\label{eq:sir}
q(f) = \sum_{S \subseteq \pset : i \in S} Q^{\id,f}(S,\emptyset,\pset \setminus S,r).
\end{equation}

\subsection{Auxiliary Proofs}
We show that disseminating events using the pseudo-random number generator approximates the dissemination model for any fanout $f$. 
Namely, for each partition $(S,I,R)$ of $\pset$, we denote by $\hist^{\id,t}(S,I,R,r)$ the set of histories $h$ of stage $t$ such that $\rnd^h = r$
and $S$ equals to the set of nodes that already received $\id$, while $I$ is the set of nodes that will forward $\id$ in the next round:
$$S = \{j \in \pset | \exists_{r'} \exists_{o \in \pset \setminus \{j\}} (\id,r') \in \re_j^h(o)\}.$$
$$I = \{j \in \pset \setminus S | \exists_{e \in \evt} (\id,e) \in \pend_j^h\}.$$

Lemma~\ref{lemma:reliable-future-equiv} demonstrates that every player $i$ expects every event to be disseminated in an approximate fashion to the dissemination model
corresponding to any fanout, in future stages.

\begin{lemma}
\label{lemma:reliable-future-equiv}
Fix any fanout $f$, stage $t \in \mathbb{N}$, terminal history $z \in \term^{t}$, player $i \in \pset$, and identifier $\id \in \ids$.
For every $r \in \{\nrndm + \id - 1 \ldots \nrnd\}$, $\xi > 0$, and partition $(S,I,R)$ of $\pset$ such that $i \in R$, there exists a function $\prng_i^f$ such that:
\begin{equation}
\label{eq:equiv}
\left|\sum_{h \in \hist^{\id,t+1}(S,I,R,r)} \pra{h}{\vsigma^f}{z} - Q^{\id,f}(S,I,R,r)\right| \leq \xi.
\end{equation}
\end{lemma}
\begin{proof}
We show the hypothesis using induction on $r$. 

\vspace{0.1in}
\textbf{Base case.} We have $r = \nrndm + \id - 1$. Fix any $h \in \hist$ such that $\pra{h}{\vsigma^f}{z} > 0$ with $\rnd^h = r$.
By construction, we have $(\id,e) \in \pend_s^h$, with $s$ being the source and $e$ the event disseminated with identifier $\id$.
Since no other node disseminates events, it holds that $\hist^{\id}(\emptyset,\{s\},\pset \setminus \{s\},r)$ includes every history leading to round $r$. Thus, it holds:
$$\sum_{h \in \hist^{\id,t}(\emptyset,\{s\},\pset \setminus \{s\},r)} \pra{h}{\vsigma^f}{z} = 1 =  Q^{\id,f}(\emptyset,\{s\},\pset \setminus \{s\},r).$$
This proves the base case.

\vspace{0.1in}
\textbf{Induction step.} Suppose that the hypothesis holds for some $r$. By construction, we have $\stat_i^h(i) = \good$ and $\miss_i^h = \emptyset$ for every $i \in \pset$ and history $h$ with $\rnd^h = r$.
Fix any history $h^1$ immediately preceding the last monitoring round of stage $t+1$. For any partition $(S,I,R)$ and $h \in \hist^{\id,t}(S,I,R,r)$ succeeding $h^1$, 
there exists exactly one history $h^*$ per set $I' \subseteq R$ succeeding $h$ where exactly the nodes from $I'$ receive $\id$ from some node in $I$ after $h$:
$$h^* \in \hist^{\id,t}(S \cup I, I', R \setminus I', r+1).$$ 
For each such $I'$, it holds:
\begin{equation}
\label{eq:rfe}
\cup_{j \in I} \prng_j^f(\seed_j^h,\id) \setminus S = I'.
\end{equation}
Since seeds are determined in the last monitoring round, if we add over every $h \in \hist^{\id}(S,I,R,r)$ where nodes of $I$ to exactly select nodes from $I'$, 
then we are adding over all the seeds selected by the mediator for nodes in $I$ after $h^1$, such that~\ref{eq:rfe} holds. 
By~PRNG1, when every seed is generated independently and uniformly at random,
the probability of such even occurring is arbitrarily close to that when nodes in $I$ select the subset of neighbours to whom they forward the
event uniformly at random. The latter event occurs with probability $\pr^{\id,f}(I' | I, R)$.
That is, an approximate fraction $\pr^{\id,f}(I' | I, R)$ of the histories from $\hist^{\id,t}(S,I,R,r)$
has nodes from $I$ selecting those from $I$. By the hypothesis, for each $I'$, there exists $\prng_i$ such that:
$$
\begin{array}{llr}
  & \sum_{h \in \hist^{\id,t}(S \cup I,I',R \setminus I',r+1)} \pra{h}{\vsigma^f}{h^1}& \approx \\
\\
\approx & \sum_{h \in \hist^{\id,t+1}(S,I,R,r)} \pra{h}{\vsigma^f}{h^1} \pr^{\id,f}(I' | I, R) & \approx\\
\\
\approx & Q^{\id,f}(S,I,R,r) \pr^{\id,f}(I' | I, R) & =\\
\\
 = & Q^{\id,f}(S \cup I,I',R \setminus I',r+1).
\end{array}
$$
Since this holds for any $h^1$, the result follows.
\end{proof}

\begin{lemma}
\label{lemma:reliable-present-equiv}
Fix any fanout $f$, stage $t \in \mathbb{N}$, player $i \in \pset$, private history $h_i \in \hist_i$, and identifier $\id \in \ids$ not 
yet introduced by the source in $h_i$. For every $r \in \{\nrndm + \id - 1 \ldots \nrnd\}$, $\xi > 0$, 
and partition $(S,I,R)$ of $\pset$ such that $i \in R$, there exists a function $\prng_i^f$ such that:
\begin{equation}
\label{eq:equiv2}
\sum_{h \in \hist^{\id,t}(S,I,R,r)} \pra{h}{\vsigma^f,\mu}{h_i}  - Q^{\id,f}(S,I,R,r) \leq \xi.
\end{equation}
\end{lemma}
\begin{proof}
The proof is identical to Lemma~\ref{lemma:reliable-future-equiv}, except that now $i$ has acquired some knowledge
regarding the outputs of $\prng_j^f$ for each $j \neq i$, after the observation of $h_i$.

\vspace{0.1in}
\textbf{Base case.} Identical to Lemma~\ref{lemma:reliable-future-equiv}.

\vspace{0.1in}
\textbf{Induction step.} Suppose that the hypothesis holds for some $r$. For any partition $(S,I,R)$, fix any $h \in \hist^{\id,t}(S,I,R,r)$ such that $\pra{h}{\vsigma^f,\mu}{h_i} > 0$,
and fix any $I' \subseteq R$. Fix $\vec{S} = (S_{\id'}^j)$, with one subset $S_{\id'}^j$ per node $j \neq i$ and identifier $\id'$ such that $|S_{\id'}^j| = f$
and $S_{\id'}^j = \prng_j^f(\seed_j^h,\id')$. Fix $h^1$ preceding $h$ and immediately preceding the last monitoring round from stage $t$.
Since the actions observed between $h^1$ and $h$ do not depend on $\prng_j^f(\seed_j^h,\id)$, for each $(L_{\id}^j)_{j \in I}$ with $|L_{\id}^j| = f$,
there exists a history $h' \in \hist^{\id,t}(S,I,R,r)$ such that $\pra{h}{\vsigma^f,\mu}{h^1} = \pra{h'}{\vsigma^f,\mu}{h^1}$ and $h'$ may differ from $h$
exactly in that $\prng_j^f(\seed_j^{h'},\id) = L_{\id}^j$ for each $j \in I$. In particular, this holds for any history $h'$ such that every node in $I'$
receives $\id$ from some node in $I$, i.e.:
\begin{equation}
\label{eq:rfe2}
\cup_{j \in I^*} \prng_j(\seed_j^{h'},\id) \setminus S = I',
\end{equation}
where $I^* \subseteq I$ is the largest subset such that for every $j \in I^*$ we have $\stat_j^{h}(j) = \good$ and $\seq(\id) \notin \miss_j^h$.
If we add over every such history $h'$, then we are adding over all the seeds for nodes in $I^*$ such that~\ref{eq:rfe2} holds.
By PRNG1, when every seed among the set that yields $S_{\id'}^j$ for every $\id' \neq \id$ is generated uniformly at random, 
this is approximately close to nodes in $I^*$ selecting the subset of neighbours to whom they forward the event uniformly at random, with an error lower than $\xi$.
This event occurs with probability $\pr^{\id,f}(I' | I^*, R) \leq \pr^{\id,f}(I' | I, R)$, since $I^* \subseteq I$. For every such $h'$, the history immediately
succeeding belongs to $\hist^{\id,t}(S \cup I,I',R \setminus I',r+1)$. Since the above holds for any $h$, an approximate fraction $\pr^{\id,f}(I' | I^*, R)$ 
of the histories from $\hist^{\id,t}(S,I,R,r)$ has the nodes from $I^*$ forwarding $\id$ to those from $I'$. Hence, we can write:
$$
\begin{array}{llr}
  & \sum_{h \in \hist^{\id,t}(S \cup I,I',R \setminus I',r+1)} \pra{h}{\vsigma^*}{z}& \leq \\
\\
\leq & \sum_{h \in \hist^{\id,t}(S,I,R,r)} \pra{h}{\vsigma^*}{z} (\pr^{\id,f}(I' | I^*, R) + \xi) & \leq\\
\\
\leq & \sum_{h \in \hist^{\id,t}(S,I,R,r)} \pra{h}{\vsigma^*}{z} (\pr^{\id,f}(I' | I, R) + \xi).
\end{array}
$$

By the hypothesis, for each $I'$, there exists $\prng_i^f$ such that for an arbitrarily small $\xi' > 0$:
$$
\begin{array}{llr}
  & \sum_{h \in \hist^{\id,t}(S \cup I,I',R \setminus I',r+1)} \pra{h}{\vsigma^*}{z}& \leq \\
\\
\leq & (Q^{\id,f}(S,I,R,r)+ \xi')( \pr^{\id,f}(I' | I, R) + \xi)& =\\
\\
\leq & Q^{\id,f}(S,I,R,r) \pr^{\id,f}(I' | I, R) + &\\
  & + \xi'(\xi +\pr^{\id,f}(I' | I, R) + Q^{\id,f}(S,I,R,r))& =\\
\\
 = & Q^{\id,f}(S \cup I,I',R \setminus I',r+1) + \xi'',
\end{array}
$$
where:
$$\xi'' = \xi'(\xi +\pr^{\id,f}(I' | I, R) + Q^{\id,f}(S,I,R,r)).$$
Since $\xi''$ is also arbitrarily small, the result holds.
\end{proof}

\end{document}